\documentclass[10pt]{article}    % Specifies the document style.  Here 

\usepackage{amssymb,latexsym,amsmath} % See Gratzer: "Math Into \LaTeX"
\usepackage{amsthm}                   % Gratzer p 139. 
\usepackage[mathscr]{eucal}  % See Gratzer P. 192.
\newfont{\sffl}{msbm10 at 16pt} % See PC TeX User Manual p.112 & p. 117
\newfont{\sff}{msbm10 at 10pt}
\theoremstyle{plain}
\newtheorem{theorem}{Theorem} % This allows the theroem enviroment,
\newtheorem{lemma}{Lemma}     %Gratzer p 140

\newcommand{\R}[1]{${\mbox{\sff R}}^{#1}$} % this produces the same result as the more common 
\newcommand{\mR}[1]{{\mbox{\sff R}}^{#1}}

\newcommand{\mRR}{{\mbox{\sff R}}}
\newcommand{\C}{${\mbox{\sff C}}$}
\newcommand{\mC}{{\mbox{\sff C}}}
\newcommand{\eq}{:=}

\begin{document}           % End of preamble and beginning of text.
\title{Uniqueness Theorems for Point Source Expansions:\\
 DIDACKS V\thanks{\small{Approved for public release; distribution is unlimited.}}}

\author{Alan Rufty}         % Declares the author's name.
\date{November 28, 2007}

\maketitle                 % Produces the title.

\newcommand{\KD}{K_{\text{D}}}

\begin{abstract}% See Gratzer: "Math Into \LaTeX", p. 46. Shoot for < 251 words.
  In the \emph{Principia Mathematica} Sir Isaac Newton proved that concentric mass shells with  equivalent mass distributions produce the same external gravitational field and thus that the problem of estimating a continuous interior mass distribution from external field information alone is ill-posed.   What is generally less well known is that finite collections of point masses contained in some bounded domain produce a unique field in the exterior domain, which means that the associated basis functions (often called ``fundamental solutions'') are independent.  A new proof of this result is given in this paper that can be generalized to other finite combinations of point source distributions.  For example, one result this paper shows in $\mathbb{R}^3$ is that a finite combination of (gravitational or electrostatic) point dipole sources contained in some interior region produces a unique field in the corresponding exterior region of interest. 

  Since no direct proofs of uniqueness results of this type are known for the \R{3} setting, which is the setting of primary practical interest, an indirect strategy is necessary.  The strategy employed in the paper is to develop results for analytic functions in the complex plane, \C, where logarithmic source basis functions correspond to point mass basis functions, and then carry them over to harmonic functions in the real plane, \R{2}, and from there to harmonic functions in \R{3}.  Although some of the results obtained can be generalized from \R{3} to \R{n}, for $n > 3$, far more results are shown for \R{2} and for \C\ than are shown for these more general settings.  For example, in the complex plane, the paper shows that a finite combination of higher order poles of any order in the interior of a unit disk always corresponds to a unique analytic function in the exterior of a unit disk. 
\end{abstract}

\vskip .05in
\noindent
\begin{itemize}
\item[\ \ ] \small{\textbf{Key words:} {Laplace's equation, inverse problem, potential theory, point sources, Vandermode\\ \phantom{Key words. L} matrix,} 
fundamental solutions, multipole}
\item[\ \ ] \small{\textbf{AMS subject classification (2000):} {Primary 86A20. Secondary 35J05, 30E05, 86A22}}
\end{itemize}

%Next define a command that represents subsections.  Here if \subsection itself is to be
% used simply replace the following with \newcommand{\SubSec}[1]{\subsection {#1}}:
\newcommand{\SubSec}[1]{

\vskip .18in
\noindent
\underline{{#1}}
\vskip .08in

}
% This is the end of the \SubSec command definition.
\newcommand{\eiii}{$\mbox{\sff E}^3$} % maybe $\mathbb{E}$ should be used--see Gratzer p 193.
\newcommand{\ls}{\vphantom{\big)}} % this lowers the subscript on inner products right paren: ).
 %to see the effects enter $$(f\,,g)_E = (f\,,g){\!\vphantom{\big)}}_{E} = (f\,,g)_{\vphantom{\big[} E} = 
 % (f   \,,g)_{\vphantom{\big[}\!E} = (f\,,g)_{\vphantom{\big)}\!E}$$
\newcommand{\lsm}{\!\vphantom{\big)}} % this is the same as \ls above except closer to the ).

% Next define a small subscript for those rare occasions when the usual one is too big:
\newcommand{\smallindex}[1]{\text{\raisebox {1.5pt} {${}_{#1}$}}}
%Example of how to use this command: $$A_{\smallindex{[\sigma]}}$$ versus $$A_{[\sigma]}$$

% The following command sets up a Large environment in Math mode so that larger integrals, etc. % are available.
%The following does the same with for a general argument:
\newcommand{\mLarge}[1]{\text{\begin{Large} $#1$ \end{Large}}}
% i.e., $$\mLarge{\iiint\limits_a^b}$$ produces a large integral sign with limits a and b.
%text size environments in ascending (descending) order:
%\normalsize, \large, \Large, \LARGE, \huge, \Huge   (Gratzer P. 275)
%\normalsize, \small, \Small or \footnotesize, \SMALL or \scriptsize, \tiny, \Tiny
%For example \text{\begin{large} stuff \end{large}}
%For different delimeters in math mode:
%\big, \Big, \bigg, \Bigg (Gratzer p.170) (example: $ big( $)
\newcommand{\mSmall}[1]{\text{\begin{footnotesize} $#1$ \end{footnotesize}}}

% Often in what follows a bigger left bracket (of specific size bigg: \bigg[ ) would be
% better than the same size brace ( \bigg{\lbrace} or \bigg{\{} ) or parenthesis (\bigg(), but
% the PC Tex version used (PCTex Version 4.0) does not have this left bracket font (it has the
% corresponding right bracket font though (and all the other size bracket fonts) nor does a
% substitute font seem to work here (at least the ones supplied with PC Tex Version 4.0 that 
% were tried didn't).  Two work arounds were found:
% First work around commands (and example compared to right bracket):
% \newfont{\cmodt}{cmr9 at 25pt}
% $$ \bigg]\!\! \raisebox{-.87ex}{{\mbox{\cmodt [}}}$$
% see Lamport LaTex 2e book p. 107 for \raisebox.
% Second work around commands (and example compared to right bracket):
% $$\bigg]\!\!\rlap{\bigg{\lceil}}\bigg{\lfloor}$$
% see The TeXbook by Knuth {1970 (first edition)} p. 82 for \rlap and Gratzer p 464 for \lceil
% and \lfloor.
% The first work around looks OK, but the resulting [ symbol has a thick line.
% Under magnification on can tell that the second work around is dead on, so it is used as 
% follows:
\newcommand{\Blbrac}{\rlap{\bigg{\lceil}}\bigg{\lfloor}} %As noted the version of PC Tex used 
% does not have a \large[, if one is available set this command to {\lbrack} (if desired).
\newcommand{\Brbrac}{\bigg{\rbrack}} %This definition is used to make it easy to change ]
% to match [. Example: $$\bigg] \Brbrac \Blbrac $$

%\mathcal{D}, \mathscr{D}, \mathbb{D}, \boldsymbol{(}$ See Gratzer pp 192-195.
% Note the following give the same result: $\mbox{\sff C}$ $\mathbb{C}$

\newcommand{\Mop}{\mathbb{M}}

\renewcommand {\baselinestretch}{1.35} % See Gratzer p. 105
 
\section{Applied Backdrop}\label{S:Backdrop}

Inverse source problems associated with harmonic functions (i.e., ones that satisfy Laplace's equation) are a  small area of modern inverse source theory; however, this area encompasses problems in geophysics, geoexploration and electrostatics, as well as many other applied sciences.  Moreover this area undoubtedly contains the oldest substantial mathematical result of any real significance associated with inverse source theory: Sir Isaac Newton showed in his \emph{Principia Mathematica} that a continuous spherically symmetric distribution produces an external field equivalent to that of a point mass located at the sphere's center, provided that the total masses of both are the same.  From this result, it follows immediately that the problem of determining continuous mass distributions inside some interior region from external gravity field information alone is an ill-posed problem, because many solutions are possible.  What is generally less well known is a proof exists for \R{n} ($n \geq 2$) that shows that a finite set of point masses located inside some bounded region produces a unique field in the exterior region \cite{geoPMuniq}.  This paper presents an alternative proof of this \R{n} result, as well as various generalizations of it. One of these generalizations is that a finite set of point dipoles is shown to be independent in \R{3}.  Other generalizations of this result are shown in the complex plane, $\mathbb{C}$, and the real plane, \R{2}.  Finally, one open question is analyzed extensively here; namely, are linear combinations of \R{3} combined point mass/point dipole basis functions linearly independent?  In the end, the independence of these basis functions is shown to hinge on the invertibility of a matrix with complex entries that is related to the Vandermode matrix [and which is specified by (\ref{E:Cmatrix2})].  Clearly, in general, this generalized Vandermode matrix is invertible, however, currently no proof of this invertibility for all $n$ is known.

   In what follows, the term point source will be  used to denote a general source term [i.e., a point mass or (point charge), a point mass dipole (or an electrostatic dipole), a point quadrapole or a higher order (electrostatic) multipole in \R{n} ($n \geq 2$)---or a logarithmic pole (with a specific form), a simple pole or a higher order pole in $\mathbb{C}$].     

 Because uniqueness results may seem to be somewhat removed from the applications arena, before considering the paper's basic approach and content in more detail, it is appropriate to briefly consider the motivations for studying point source uniqueness results.  The primary motivations can be framed as follows:
\begin{enumerate}
\item
By better understanding and quantifying various point source uniqueness results, the hope is that geophysical inverse source theory can be better understood and placed on a firmer physical and mathematical foundation.  (Point source uniqueness results should be viewed as merely a first tentative step in this goal.)
\item
 Point source uniqueness results show that, in theory, when well formulated algorithms are properly implemented, point source determination software always produces reliable results.  In particular, by showing that combinations of point sources and point dipoles are linearly independent, one can immediately infer that the DIDACKS implementation that interpolate for the  scalar potential of gravity and the vector components of gravity are mathematically well framed, as discussed in \cite{DIDACKS,DIDACKSI,DIDACKSII} and \cite{DIDACKSIII}.
\item
 As discussed in Appendix A of \cite{DIDACKSII}, this uniqueness result for point sources and point dipoles also shows that geophysical collocation procedures for similar data sets (such as geoid height and the vector components of gravity disturbance) are mathematically consistent since the associated error-free covariance matrix can always be inverted.  
\item
  Finally, as always, when a different topic is addressed, different mathematical techniques may be called for, and thus there is the possibility of uncovering new mathematical theorems, results and techniques that may be of interest in other domains.  Here, not only are the results dealing with logs and poles in \C\ suggestive, but so are several of the (geometric) theorems that allow for results to be carried over from \R{2} to \R{n}.  Also, since Vandermode's matrix in an extended form is used here, it is not unreasonable to think that there are mathematical links to general interpolation theory, where the Vandermode matrix is often employed, and that this is a two-way street (which means that there may be general theorems in interpolation theory that allow for more general point source uniqueness theorems than those proven here).
\end{enumerate}
Here item 1. is worth delving into a little deeper, since placing inverse source gravity theory (i.e., $\mathbb{R}^3$ Laplacian inverse source theory) on a firmer mathematical and physical foundation is one goal of the DIDACKS sequence of papers initiated by \cite{DIDACKSI}.   In this regard, there are a number of interconnected points that are relevant, and it is appropriate to point one of them out here.  First, consider the point made in \cite{Featherstone} that since each point source can be replaced by an equivalent sphere of uniform density, one can replace a distribution of point sources with an equivalent set of spheres that produces a more uniform distribution of mass.  Thus, consider a collection of point masses that have different strengths from point to point and are located on a regular cubic grid.  Since the grid spacing is uniform each point mass can be replaced by a sphere of like size that has a corresponding density.  From the uniqueness theorems given here one can infer that the estimation problem for the associated densities of this spherical arrangement is well posed.  Moreover, since it is well known that as the grid spacing for a set of point masses increases the underlying estimation problem becomes better conditioned, one can immediately infer that the same holds for this spherical approximation.  This result has other generalizations and interpretations as well.  For example, in geoexploration it is common to use a set of right parallelepipeds of uniform density to approximate continuous mass densities since they produce potentials that can be calculated in closed form \cite[p. 398]{Pick} and they yield inverse source density estimation equations that are (relatively) stable.  As a next step, one might consider extending the uniqueness results here to uniform cubic latices (or one might consider applying DIDACKS theory to the estimation of such cubic distributions).  This and other  points raised above clearly dovetail with similar discussion in \cite{DIDACKSIV}.

\section{Mathematical Preliminaries}\label{S:intro}

 The general method of proof employed is to start with the easiest to handle uniqueness case---which turns out to be a distribution of simple poles of the form $1/(z-z_k)$ in the complex plane--and then generalize this result in various ways:  First to uniqueness of logarithmic basis functions in \C, then to higher order poles of the form $1/(z-z_k)^n$.  Next a (well-known) one-to-one correspondence between various types of expansions in \C\ and \R{2} is noted: logarithmic potentials $\iff$ point masses;  simple poles $\iff$ dipoles; poles of order $n$ $\iff$  multipoles of higher order for $n = 1,\,2,\,3$.  It is then shown that uniqueness of an expansion in \C\ implies uniqueness of a like expansion in \R{2}---thus proving the \R{2} cases.  Finally it is shown explicitly that point mass uniqueness results in \R{2} imply point mass uniqueness results in \R{3}.  These correspondences and the symbols used to represent the associated expansions in the sequel are shown in Table~\ref{Ta:Sources}.  The arguments employed can clearly be generalized in two ways:  (1) To show that similar results hold for multipoles of order three or less in \R{3}.  (2) To show that results obtained in \R{3} imply similar results in \R{N} for $N \geq 2$.  For concreteness and for ease of exposition these two generalizations will not be explicitly dealt with here since they do not correspond to commonly occuring physically relevant cases and, when handled explicitly, entail some technical difficulties.

  In \R{2} and \R{3} the conservative vector fields of interest can be obtained by taking the gradient of a scalar field:   
\begin{equation}\label{E:grad}
 \vec{F}(\vec{X}) = - \mathbf{\nabla} U(\vec{X})\,,
\end{equation}
where, of course, $\mathbf{\nabla}$ is the $N-$dimensional gradient and where $\vec{X} \in \mR{N}$.  Potentials specified by (\ref{E:grad}) satisfy Laplace's equation in $N-$dimensions: ${\nabla}^2 U = 0$.  Arguably the most significant cases, and the ones that will be focused on here, are the electrostatic and gravitational fields in \R{2} and \R{3}.  Because of the issue of the sign associated with the mutual attraction of point sources and the issue of the sign of the energy density play no role in what follows, the only relevant difference between the gravitational and electrostatic cases is that point mass strengths ($m_k$) are generally assumed to be non-negative, but electric change strengths ($q_k$) can have either sign.  Since point source uniqueness results that hold for sources of either sign also clearly hold when the sources are all assumed to be positive and since a historical precedent exists for stating uniqueness results in terms of point masses, in the sequel it will always be assumed that while the scalar point source parameters are called a point masses and denoted $m_k$, their assigned values can take on either sign.  Likewise the vector point source distribution that corresponds to a electrostatic dipole or point mass dipole will be denoted ${\vec{D}}_k$ and it will be called a point dipole.  Higher order point multipoles will also be considered in what follows and point masses, point dipoles and point multipoles will be collectively referred to as point sources.  In all cases, since either sign is allowed for these point sources and since uniqueness results pertain to whether the source terms form a linearly dependent set of basis functions in the exterior region of interest, the overall choice of sign convention for these basis functions and associated point source strengths does not matter here and the basis sign conventions are chosen for convenience.  As previously noted, while point source  distributions for \R{N} for $N > 3$ exist and the results here can be naturally extended to handle them they will not be explicitly considered here.  Finally, observe that magnetic dipoles have a different form from the gravitational or electrostatic (and magnetostatic) potential forms assumed here and the question of adapting the results here to their study will also not be addressed.

  First, as a general convention, let $N_k$ denote a finite integer greater than zero and let the subscript $k$ (or $k'$) be used to index the $N_k$ point sources under consideration so that $k$ and $k' = 1,\,2,\,3,\,\ldots,\,N_k$ is always understood.   To further fix notation in \R{N} for $N \geq 2$, let the $N_k$ fixed distinct point-source locations be specified by ${\vec{X}}_k$ so that ${\vec{X}}_k \neq {\vec{X}}_{k'}$ when $k' \neq k$.  It will also be assumed that all of the sources are located in some bounded domain.

  Introducing a specific symbol for point mass potentials, $V$, instead of the general symbol $U$ used in (\ref{E:grad}) results in the following definition 
\begin{equation}\label{E:PM2}
V(\vec{X}) \equiv \sum\limits_{k=1}^{N_k}\ m_k\,\ln\ (|{\vec{X}} - {\vec{X}}_k|^{-1})
\end{equation}
in \R{2} and 
\begin{equation}\label{E:PM3}
V(\vec{X}) \equiv \sum\limits_{k=1}^{N_k}\ \frac{m_k\ \ }{|{\vec{X}} - {\vec{X}}_k|}
\end{equation}
in \R{3}. 

   Likewise the point dipole potential form is
\begin{equation}\label{E:DP2}
W(\vec{X}) \equiv \sum\limits_{k=1}^{N_k}\ {\vec{D}}_k\mathbf{\cdot}\,{\mathbf{\nabla}}\,\,\ln\ (|{\vec{X}} - {\vec{X}}_k|^{-1})
\end{equation}
in \R{2} and
\begin{equation}\label{E:DP3}
W(\vec{X}) \equiv \sum\limits_{k=1}^{N_k}\ {\vec{D}}_k\mathbf{\cdot}\,{\mathbf{\nabla}}\,({|{\vec{X}} - {\vec{X}}_k|^{-1}})
\end{equation}
in \R{3}.  The relationships of these potentials and their corresponding complex analytic counter parts are shown in Table~\ref{Ta:Sources}.  These correspondences will be addressed below.
%Next is from Hahn Section 5.2 etc.; Gratzer section 3.7:
%\vskip .2in
\begin{table} %note that \begin{table}[b] puts the table at the bottom of the page (see Gratzer).
% For multirow in the following see Gossens et. all (LateX Companion) p 135, but is 
% not defined here for some reason.
\begin{center}
\begin{tabular}{|c|c|||c|c|}\hline
\R{N} Description   &  \R{N} Symbol & \C\ Analog         & \C\ Symbol   \\ \hline\hline
 Harmonic Functions &     $U$       & Analytic Functions &    $f$       \\ \hline
 Point Masses       &   $\!V$       & Logarithm Terms    &    $g$       \\ \hline
 Point Dipoles      &     $W$       & Simple Poles       &    $h$       \\ \hline
 Multipoles         & $\ \ H^{(n)}$ & Higher Order Poles &    $h^{(n)}$ \\ \hline
\end{tabular}
\caption{Corresponding Types of Potential and Analytic Functions}\label{Ta:Sources}
\end{center}
\end{table}
%\vskip .3in

  Next consider uniqueness results stated in terms of scalar potentials.  So long as Laplace's equation is satisfied in the region of interest the exact shape of the exterior region is immaterial due to the uniqueness of Dirichlet boundary value problems so, without loss of generality, assume that the mass distributions are located in some bounding sphere.  Further, since the origin of coordinates and the length scale also does not change, for the desired final uniqueness results it can be assumed that the harmonic region of interest is $|{\vec{X}}| \geq 1$ and that the sources are in the compliment of this region: $0 < |{\vec{X}}_k| < 1$ (where, for later convenience, it is assumed that the origin is not situated directly over any particular source).  Point mass uniqueness requirements can then be stated as the condition that if $V(\vec{X}) = 0$, for all ${\vec{X}} \geq 1$, then $m_k = 0$ for all $k$ and that the converse also holds.  Which is to say that for any finite $N_k > 0$, if $m_k \neq 0$, for all $k$, then $V(\vec{X}) \neq 0$ for some values of ${\vec{X}} \geq 1$.  Dipole uniqueness can be stated similarly:  If $V(\vec{X}) = 0$, for all $|{\vec{X}}| \geq 1$, then ${\vec{D}}_k = 0$ for all $k$, or conversely as the condition that for any finite $N_k > 0$ if ${\vec{D}}_k \neq 0$, for all $k$, then $V(\vec{X}) \neq 0$.  It is assumed throughout that any nonzero values of $m_k$ and ${\vec{D}}_k$ are bounded (while the fixed nature of ${\vec{X}}_k$ rules out the formation of dipoles in the limit $m_k \rightarrow \infty$ and $m_{k'} \rightarrow -\infty$ with $|m_k| = |m_{k'}|$ for some pair of points indexed by $k$ and $k'$, the unbounded mass case is still best bypassed).

  Several observations are relevant.  By introducing the $\mathbb{R}^2$ basis functions 
\begin{equation}\label{E:PsiK}
{\Psi}_k(\vec{X}) = \ln\,\frac1{|{\vec{X}} - {\vec{X}}_k|} - \ln\,\frac1{|{\vec{X}}|} =  \ln\,\frac{|{\vec{X}}|}{|{\vec{X}} - {\vec{X}}_k|}\ ,
\end{equation}
 Equation~(\ref{E:PM2}) can be reexpressed as  
\begin{equation}\notag
V(\vec{X}) = \sum\limits_{k=1}^{N_k}\ m_k{\Psi}_k(\vec{X})\ +  m_0\,\ln\ \frac1{\vec{X}} 
\end{equation}
where $m_0 = \sum_{k=1}^{N_k} m_k$.  The case $m_0 \neq 0$ can be easily disposed of since ${\Psi}_k \rightarrow 0$ as ${|\vec{X}| \to \infty}$ and it is clear that $|V| \rightarrow \infty$ as ${|\vec{X}| \to \infty}$ unless $m_0 = 0$.  Consequently, without loss of generality take $m_0 = 0$, so that the form 
\begin{equation}\label{E:PMPsi}
V(\vec{X}) = \sum\limits_{k=1}^{N_k}\ m_k{\Psi}_k(\vec{X})
\end{equation}
is always assumed for $\mathbb{R}^2$ in what follows.  Likewise all of the harmonic functions considered here will be assumed to vanish at infinity by convention.

  In $\mathbb{R}^N$, uniqueness results for a scalar potential imply uniqueness results for the vector field itself as can be seen from the following argument.  Consider the line integral
\begin{equation}\label{E:LineInt}
  U(\vec{X}) = U({\vec{X}_o}) + \int_{\vec{X}_o}^{\vec{X}} {\mathbf{\nabla}}U\mathbf{\cdot}\,d\,\vec{\ell}
\end{equation}
where $\vec{\ell}= \vec{\ell}(s)$ denotes a parameterized path (with arclength $s$).  The contention is that since $\vec{X}$ and ${\vec{X}_o}$ are arbritrary points in the exterior region and $\vec{\ell}(s)$ is also assumed to lie wholly in this exterior region, $\vec{F} \neq 0 \iff V \neq 0$ and conversely $\vec{F} = 0 \iff U = 0$, where $\vec{F}$ and $U$ are related by (\ref{E:grad}). For example, if $\vec{F}(\vec{X}') < 0$ at some point $\vec{X}' > 1$, then due to the mean value theorem for harmonic functions $\mathbf{\nabla} U > 0$ holds in some finite neighborhood of $\vec{X}'$ so that both $\vec{X}$ and ${\vec{X}_o}$, along with the path connecting them in (\ref{E:LineInt}), can be taken to be inside this same neighborhood.  This, in turn, means that at least one of the two values $U(\vec{X})$ or $U({\vec{X}_o})$ must be nonzero. Alternatively, if $U(\vec{X}) > 0$, then ${\vec{X}_o}$ can be set to a point at infinity and from  (\ref{E:LineInt}) it is clear that $\vec{F} \neq 0$ must occur at someplace along the line integral.  Without loss of generality, uniqueness results will thus be stated in terms of scalar potentials for convenience.

  Uniqueness in the complex setting is addressed first since this is the easiest route to the desired \R{2} results, which can be readily generalized to \R{3}.  It is useful to have a common (and commonly used) symbolism for addressing uniqueness issues in \R{2} and \C.  Let $x$ and $y$ denote standard Cartesian coordinates in either setting: In \R{2}, ${\vec{X}} \equiv (x,\,y)^T$ and ${\vec{X}}_k \equiv (x_k,\,y_k)^T$ (where $T$ denotes the transpose); while in \C, $z = x + i\,y$ and $z_k = x_k + i\,y_k$.  In both settings $\sqrt{x^2 + y^2} \geq 1$ and  ${\sqrt{x_k^2 + y_k^2}} < 1$.  Recall from elementary treatments of analytic functions that there is a general mapping between harmonic functions defined over some subregion of \R{2} and analytic functions defined over the corresponding subregion of the complex plane as indicated by Table~\ref{Ta:Sources}.  This mapping can be done uniquely when certain reasonable conditions are met with regards to branch-cuts and the nature of the region under consideration and it is assumed that the reader is familiar with them.  Specifically, if $U$ is harmonic in \R{2} let $f(z)$ denote the unique analytic function in \C\ whose real component corresponds to $U(x,\,y)$, in which case $f(z)$ will be called the standard completion of $U(x,\,y)$.  The standard completion of the sum of real functions obeys the principle of linear superposition, so that, for example, the standard completion of $V$ is a linear superposition of terms that are the standard completions of the ${\Psi}_k(\vec{X})$.  Specifically, let standard completions of ${\Psi}_k(\vec{X})$ be denoted ${\psi}_k(z)$.  Since the real part of  $\ln z$ is $\ln |z|$, it is obvious from (\ref{E:PsiK}) that the standard completion of ${\Psi}_k$ is given by
\begin{equation}\label{E:psi}
 {\psi}_k(z) \, = \,\ln\,\frac z{(z - z_k)}\, \eq  \,\ln\,\frac1{z - z_k} - \,\ln\,\frac1{z}
\end{equation}
and thus (using the notation indicated in Table~\ref{Ta:Sources} for the corresponding complex logarithmic case) $g(z)$ can be written as
\begin{equation}\label{E:complexLogs} 
  g(z) \equiv \sum\limits_{k=1}^{N_k} {\mu}_k\, {\psi}_k(z)\ , 
\end{equation}
where ${\mu}_k \in \mC$.

  As before, a general series based on $\ln\,[1/{z - z_k}]$ can be contemplated here since it corresponds to adding a term involving $1/z$ to (\ref{E:complexLogs}); however, it is readily proved that when this term is present it always causes $g(z)\neq 0$ for large $z$ and so its presence need not be considered further.  This might all seem straightforward, but even here some care is called for.  Moreover, since the basic strategy used in the sequel will be to prove uniqueness (i.e., linear independence) in one setting and then to obtain uniqueness results in all other settings by applying a linear uniqueness preserving mapping, the general theorems that are required are best stated explicity (for clarity and uniformity of exposition).   Even though the underlying concepts are well known, these linear independence preserving theorems do not necessarily take conventional forms.
 
\section{Uniqueness Preserving Mappings}\label{S:Umaps}

  This section discusses uniqueness preserving mappings that are used in the sequel.  As noted at the end of the last section, the first of these mappings is associated with the act of standard analytic completion and entails a unique correspondence between analytic functions in the complex pane and harmonic functions in \R{2}.  Since uniqueness results will first be shown in the complex setting and then mapped into the \R{2} setting, first consider what uniqueness means in the complex setting:

\vskip 9pt

\noindent
{\bf{Definition \ref{S:Umaps}.1}}\ \ 
A set of basis functions ${\{f_k\}}_{k=1}^{N_k}$ is said to produce a \emph{unique} expansion in the complex setting when they meet the following criteria.  First each basis function must be a bounded analytic function for $|z| \geq 1$.  Second each basis function must vanish at infinity.  Third, a sum of the form 
 \begin{equation}\label{E:Fseries}
f(z) = \sum\limits_{k=1}^{N_k} {\mu}_k f_k(z)\ \ \text{with}\ \ {\mu}_k \in C
\end{equation}
 must be linearly independent for $|z| \geq 1$, where linear independence means that $f(z) = 0$ holds if and only if ${\mu}_k = 0$ for all $k$.

\vskip 9pt

Let ${\text{Re}}\,\{f\}$ denote the real part of $f$ and ${\text{Im}}\,\{f\}$ the imaginary part, then if $u_k(x,\,y) \equiv {\text{Re}}\,\{f_k(z)\}$, $v_k(x,\,y) \equiv {\text{Im}}\,\{f_k(z)\}$, ${\alpha}_k \equiv {\text{Re}}\,\{{\mu}_k\}$ and ${\beta}_k \equiv {\text{Im}}\,\{{\mu}_k\}$; from $f_k(z) = u_k(x,\,y) + i\,v_k(x,\,y)$  and ${\mu}_k = {\alpha}_k  + i\,{\beta}_k$ it follows that
\begin{equation}\label{E:ReF}
{\text{Re}}\,\{f\} = \sum\limits_{k=1}^{N_k} {\text{Re}}\,\{{\mu}_k f_k(z)\}\ = \sum\limits_{k=1}^{N_k} [{\alpha}_k\,u_k(x,\,y) - {\beta}_k\,v_k(x,\,y)].
\end{equation}
Since ${\mu}_k = 0$ for all $k$ implies that both ${\alpha}_k = 0$ and ${\beta}_k = 0$ hold for all $k$, uniqueness of the analytic set of basis functions ${\{f_k\}}_1^{N_k}$, implies simultaneous uniqueness of the pair of conjugate harmonic basis function sets ${\{u_k(x,\,y)\}}_{k=1}^{N_k}$ and ${\{v_k(x,\,y)\}}_{k=1}^{N_k}$.  Uniqueness in the \R{2} setting is defined analogously to Definition~\ref{S:Umaps}.1. 

 This result can be summarized as a theorem:
\begin{theorem}\label{T:CtoRmap}
If the set of basis functions ${\{f_k\}}_{k=1}^{N_k}$ are unique in the complex setting and if $u_k(x,\,y) \equiv {\text{Re}}\,\{f_k(z)\}$ and $v_k(x,\,y) \equiv {\text{Im}}\,\{f_k(z)\}$, then the combined set of basis functions ${\{u_k(x,\,y)\}}_{k=1}^{N_k}\cup{\{v_k(x,\,y)\}}_{k=1}^{N_k}$ are linearly independent or unique in the \R{2} setting. 
\end{theorem}
\noindent
Notice that since the act of standard harmonic completion is unique, if uniqueness can be shown in \R{2} for a sequence of harmonic conjugate pairs then uniqueness for basis functions of the form (\ref{E:Fseries}) follows. 

For the complex setting, in what follows only two general types of basis functions will need to be considered: (1) Logarithmic potentials discussed above of the form ${\psi}_k(z)$ (which, as discussed latter, are analytic from the branch-cut considerations in the first part of Appendix~A).  (2) Poles of the form ${\mu}_k/(z - z_k)^n$, for finite $n > 0$.  For $f(z) = \sum_{k=1}^{N_k} {\mu}_k {\psi}_k(z)$, the restriction ${\beta}_k = 0$ can be made since the argument dependent parts of the logarithmic term occuring in (\ref{E:psi}) [i.e., ${\text{Im}}\,\{{\psi}_k(z)\}$] are not of general interest in the \R{2} harmonic setting.  Uniqueness results in the complex plane for a series of logarithmic basis functions can thus be used to show point mass uniqueness in \R{2}.

  Next consider a series of poles of fixed order $n$: 
\begin{equation}\label{E:Hseries}
h^{(n)}(z) = \sum\limits_{k=1}^{N_k} \frac{{\mu}_k}{(z - z_k)^n}\ \ \text{with}\ \ {\mu}_k \in C\ .
\end{equation}
The case $n = 1$,  $h(z) \equiv h^{(1)}(z)$, is of special interest and results in the well-known linear combination of simple poles: 
\begin{equation}\label{E:SimplePoles}
h(z) = \sum\limits_{k=1}^{N_k} \frac{{\mu}_k}{(z - z_k)} \ .
\end{equation}
Uniqueness results will first be shown for an expansion of this form.
Since 
\begin{equation}\label{E:PoleExpansion}
 \frac{{\mu}_k}{z - z_k} = \frac{{\mu}_k(z^* - z^*_k)}{|z - z_k|^2}  = \frac{{\alpha}_k(x -x_k) + {\beta}_k(y -y_k)}{{(x - x_k)^2 + (y -y_k)^2}} +  i\,\frac{{\beta}_k(x -x_k) - {\alpha}_k(y -y_k)}{{(x - x_k)^2 + (y -y_k)^2}}
\end{equation}
the real part of $h(z)$ corresponds to an \R{2} dipole expansion $W(\vec{X})$ given by (\ref{E:DP2}), as will now be explicitly shown.

  The various point dipole terms contained in (\ref{E:DP2}) are also known as first order multipoles.  For future reference it is useful to have a consistent notation for delineating multipole basis functions of various orders.  The order of a multipole corresponds to the number of subscripts it has, so a first order multipole has a single subscript that can take on the values $1,\,2,\,3,\,\ldots,\,N$ in \R{N}.  Thus the subscripts of a dipole (or first order multipole) basis function are associated with the various directions in \R{N}.  In \R{2} this subscript takes on two values so that the two first order \R{2} multipole basis functions associated with the source position ${\vec{X}}_k$ can be written as ${\Mop}_j^{{}_{[k]}} \equiv {\Mop}_j({\vec{X}},\,{\vec{X}}_k)$ for $j = 1,\,2$.  The \R2 dipole or first order multipole basis function oriented along the $x$-axis in is given by
\begin{equation}\label{E:Mk1}
{\Mop}_1^{{}_{[k]}} \equiv \frac{\partial \ }{\partial x} \ln \, \frac1{\sqrt{(x - x_k)^2 + (y -y_k)^2}} = - \frac{x -x_k}{{(x - x_k)^2 + (y -y_k)^2}}\ .
\end{equation}
  Likewise a unit dipole or multiple aligned along the $y$-axis in \R2 is given by
\begin{equation}\label{E:Mk2}
{\Mop}_2^{{}_{[k]}} \equiv \frac{\partial \ }{\partial y} \ln \, \frac1{\sqrt{(x - x_k)^2 + (y -y_k)^2}} = - \frac{y -y_k}{{(x - x_k)^2 + (y -y_k)^2}}\ .
\end{equation}
Thus uniqueness results in the complex plane for a series of simple poles can be used to show  dipole uniqueness in the \R{2} setting. Theorem~\ref{T:CtoRmap} (and its reverse) clearly involves a change of setting form \C\ to \R{2} (or \R{2} to \C)

\section{Uniqueness of Complex Poles of The Same Order}\label{S:Complex}

This section considers relatively straightforward results about uniqueness of expansions of poles and uniqueness of expansions of logarithmic basis functions.  In particular, the theorem that shows that an expansion in terms of simple poles, $h(z)$ as given by (\ref{E:SimplePoles}), is unique is readily stated and proved.  The analogous result is also easy to prove for expansions of higher-order poles of a given type (i.e., where all the poles are of some specified order), as well logarithmic basis functions.  A consideration of the more difficult uniqueness results for mixed types of expansions are postponed until Section~\ref{S:Mixed} and, in the end, concrete proofs of these mixed uniqueness results prove to be elusive.

Before proceeding to a statement of the desired theorem, several observations are in order.  First, if $z_k = 0$ for any $k$, then it is simply necessary to translate and rescale, when required, so that $z_k \neq 0$ can be assumed.  Second, as in \R{3} for concentric spheres, for uniform circular distributions of continuous simple poles in \C\ and \R{2} immediate counter examples can be constructed.  Third, in an attempt to derive the wanted uniqueness results, it is tempting to try to directly apply the standard theory of poles and residues associated with analytic function theory.  For example, applying the residue theorem by taking a closed line integral around the unit disk immediately shows that $\sum_{k=1}^{N_k} {\mu}_k = 0$; however, additional progress quickly becomes difficult since, in order to make further progress, it is necessary to consider paths that extend into the interior region, where some form of analytic continuation inside the unit disk must be used, but this is, at best, questionable.  These issues are addressed further in Appendix~A.

For simple poles the desired uniqueness theorem is:
\begin{theorem}\label{T:SmpPoleUniq}
If $h(z)$ has the form specified by (\ref{E:SimplePoles}) where $N_k$ is finite, $0 < |z_{k}| < 1$ and $z_{k'} \neq z_k$ for $k' \neq k$, then $h(z) = 0$ for all $|z| \geq 1$ if and only if ${\mu}_k = 0$ for all $k$. 
\end{theorem}
\begin{proof}
Here ${\mu}_k = 0$ for all $k$ trivially implies $h(z) = 0$ for all $z$, so only the converse needs to be considered.  Throughout the proof assume that $0 < |z_{k}| < 1$, $z_{k'} \neq z_k$ for $k' \neq k$ and that $|z| \geq 1$.  The proof will be by contradiction, so assume to the contrary that a non-unique expansion exists where ${\mu}_k \neq 0$, but $h(z) = 0$ for all $|z| \geq 1$.  If ${\mu}_k = 0$  occurs for any $k$ in this expansion, then drop out these terms, reindex the ${\mu}_k$'s and reduce the value of $N_k$, so that  ${\mu}_k \neq 0$ can be assumed to hold for all $k$ without loss of generality.  Using the geometric series, each of the pole terms appearing in (\ref{E:SimplePoles}) can be reexpressed as 
\begin{equation}\label{E:geoseries}
\frac{1}{z - z_k} = \sum\limits_{n=0}^{\infty} \frac{z_k^n}{z^{n+1}}\ ,
\end{equation}
which has the same overall form as the power series
\begin{equation}\label{E:series1}
 f(z) = \sum\limits_{n=1}^{\infty} \frac{a_n}{z^n}\,,\,\ \text{where}\,\  a_n \in \mC.
\end{equation}
As pointed out in many elementary treatments of complex variables, $f(z) = 0$ for all $z$ if and only if $a_n = 0$ for all $n \geq 1$, which is a useful condition here.  Substituting (\ref{E:geoseries}) into (\ref{E:SimplePoles}) allows $h$ to be rewritten in the form:
\begin{equation}\label{E:series2}
 h(z) = \sum\limits_{j=1}^{\infty} b_jz^{-j}\,,\ \text{where}\  b_j =  \sum\limits_{k=1}^{N_k} z_k^{j-1}\,{\mu}_k\ .
\end{equation}
By assumption $h(z) \equiv 0$ holds for some given $z_k$'s and ${\mu}_k$'s with $z_k \neq z_{k'}$ and ${\mu}_k \neq 0$ for $k = 1\,2,\,3,\,\cdots\,N_k$ with $N_k > 0$.   Also as noted above, the factors $b_j$ occuring here are unique so $b_j \equiv 0$ for $j = 0,\,1,\,2,\,3,\,\cdots\,$.  From (\ref{E:series2}) this condition on the $b_j$'s can be rewritten in matrix form as
\begin{equation}\label{E:unique}
 \widetilde{\mathbf{G}}\,{\mathbf{\mu}} = 0\ {\text{where}}\ \  
   \widetilde{\mathbf{G}} \eq
        \begin{pmatrix}
               1      & 1      & 1      & \cdots & 1\\
               z_1    & z_2    & z_3    & \cdots & z_{N_k}\\
               z_1^2  & z_2^2  & z_3^2  & \cdots & z_{N_k}^2\\
               z_1^3  & z_2^3  & z_3^3  & \cdots & z_{N_k}^3\\
               \vdots & \vdots & \vdots &        & \vdots
        \end{pmatrix}
\end{equation}
and where ${\mathbf{\mu}} = ({\mu}_1,\,{\mu}_2,\,{\mu}_3,\,\cdots,\,{\mu}_{N_k})^T$.  Consider the square matrix formed by the first $N_k$ rows and $N_k$ columns of $\widetilde{\mathbf{G}}$.  This square matrix is the Vandermode matrix, which has a determinant that is well known to be nonzero when the $z_k$ are distinct \cite{PJDavis}.  Thus the only solution to (\ref{E:unique}) is ${\mathbf{\mu}} = 0$, contrary to our original assumption, proving the uniqueness of the form given by (\ref{E:SimplePoles}).
\end{proof}

  An analogous result can easily be shown for the logarithmic form given by (\ref{E:complexLogs})
\begin{equation}\label{E:complexLogs2} 
  g(z) \eq \sum\limits_{k=1}^{N_k} {\rho}_k\, {\psi}_k(z)\ , 
\end{equation}
where $\rho \in \mathbb{C}$ is the source parameter and, as before,
\begin{equation}\label{E:psi4}
 {\psi}_k(z) \, = \,\ln\,\frac z{(z - z_k)}\ .
\end{equation}
Here, since $z_k = 0$ implies ${\psi}_k(z) = 0$, so $z \neq 0$ and $|z| \geq 1 > |z_k| > 0$ is assumed (as usual).
Notice that from the discussion given in Appendix~A, ${\psi}_k(z)$ has no branch cuts for $|z| \geq 1$, which means that it is analytic for $|z| \geq 1$ and thus that it has a proper power series representation:
\begin{equation}\label{E:PsiSeries} 
   {\psi}_k(z) = \sum\limits_{n=1}^{\infty} \frac{z_k^n}{nz^n}\ . 
\end{equation}
Although (\ref{E:PsiSeries}) can be obtained in various ways, it is easy to see that it is the correct series by simply comparing the derivative of the RHS of (\ref{E:PsiSeries}) with the series expansion of the derivative of the RHS of (\ref{E:psi4}).

 Substituting the series expansion (\ref{E:PsiSeries}) into  (\ref{E:complexLogs2}) gives 
\begin{equation}\label{E:complexLogs3} 
  g(z) = \sum\limits_{k=1}^{N_k} \,\sum\limits_{n=1}^{\infty} \frac{z_k^n{\rho}_k}{nz^n}\ , 
\end{equation}
and setting the various powers of $z^n$ to zero yields an equation set similar to (\ref{E:unique}) that must hold if $g(z) = 0$ is to hold:
\begin{equation}\label{E:uniqueG}
 \widetilde{\mathbf{G}}'\,{\mathbf{\rho}} = 0\ {\text{where}}\ \  
   \widetilde{\mathbf{G}}' \eq
        \begin{pmatrix}
               z_1    & z_2    & z_3    & \cdots & z_{N_k}\\
               z_1^2/2  & z_2^2/2  & z_3^2/2  & \cdots & z_{N_k}^2/2\\
               z_1^3/3  & z_2^3/3  & z_3^3/3  & \cdots & z_{N_k}^3/3\\
               \vdots & \vdots & \vdots &        & \vdots
        \end{pmatrix}
\end{equation}
and where ${\mathbf{\rho}} = ({\rho}_1,\,{\rho}_2,\,{\rho}_3,\,\cdots,\,{\rho}_{N_k})^T$.  As before, it is necessary to consider only the first $N_k$ rows of $\widetilde{\mathbf{G}}'$, which can be reexpressed as a simple matrix product of the Vandermode matrix $\mathbf{G}$ and two other $N_k \times N_k$ matrices. If particular, if $\mathbf{G}'$ denotes this $N_k \times N_k$ matrix, then 
\begin{equation}\label{E:uniqueG2}
 {\mathbf{G}}' \eq \mathbf{N}^{-1}\mathbf{G}\mathbf{X}\,,\ {\text{where}}\ \  
  {\mathbf{N}} \eq
        \begin{pmatrix}
               1  &  0  & 0  & \cdots & 0\\
               0  &  2  & 0  & \cdots & 0\\
               0  &  0  & 3  & \cdots & 0\\
               \vdots & \vdots & \vdots &        & \vdots\\
               0  &  0  & 0  & \cdots & N_k\\
        \end{pmatrix}
\ \ {\text{and}}\ \ 
 {\mathbf{X}} \eq
        \begin{pmatrix}
               z_k  &  0  & 0  & \cdots & 0\\
               0  &  z_2  & 0  & \cdots & 0\\
               0  &  0  & z_3  & \cdots & 0\\
               \vdots & \vdots & \vdots &        & \vdots\\
               0  &  0  & 0  & \cdots & z_{N_k}\\
        \end{pmatrix}
\,\, .
\end{equation}
 The condition $\widetilde{\mathbf{G}}'\,{\mathbf{\rho}} = 0$ thus becomes $\mathbf{N}^{-1}\mathbf{G}\mathbf{X}\,{\mathbf{\rho}} = 0$, which immediately implies ${\mathbf{\rho}} = 0$ since all three matrices involved are invertible.  This result can be restated as the desired  uniqueness of logarithmic expansions given by the form (\ref{E:complexLogs}): 
\begin{theorem}\label{T:CmpLogUniq}
If $g(z)$ has the form specified by (\ref{E:complexLogs}), where $N_k$ is finite, $0 < |z_{k}| < 1$ and $z_{k'} \neq z_k$ for $k' \neq k$, then $g(z) = 0$ for all $|z| \geq 1$ if and only if ${\mu}_k = 0$ for all $k$. 
\end{theorem}

   Next consider the issue of the uniqueness of higher-order pole expansions of the form
\begin{equation}\label{E:Pole1}
h^{(m)}(z) \eq \sum\limits_{k=1}^{N_k} \frac{{\nu}_k}{(z - z_k)^{m}}
\end{equation}
where $m$ is a positive integer, $\nu_k \in \mathbb{C}$ and the usual restrictions apply to $z$ and $z_k$.  A series expansion for $(z - z_k)^{-m} = z^{-m}(1 - z_k/z)^{-m}$ can be found directly from the binomial series
\begin{equation}\label{E:bimomial}
\frac{1}{(1 - x)^{m}} = 1 + mx + \frac{m(m + 1)}{2!}x^2  + \frac{m(m + 1)(m + 2)}{3!}x^3 + \cdots +  \frac{m(m + 1)\cdots(m + n - 1)}{n!}x^n + \cdots \ .
\end{equation}
The general coefficients here are related to the usual binomial coefficients.  Introducing the special (and non-standard) symbol $S_{m,\,n}$ for these signed coefficients and replacing $x$ with $z_k/z$ thus produces 
\begin{equation}\label{E:BinCoef}
\frac{1}{(z - z_k)^{m}} = \frac1{z^m}\,\sum\limits_{n=0}^{\infty} S_{m\,n}{\left(\frac{z_k}{z}\right)}^n\ \ \text{where}\ \ S_{m,\,0} \eq 1\ \ \text{and}\ \ S_{m,\,n} \eq \frac{(m + n - 1)!}{n!(m - 1)!}\ \text{for}\ n > 0\,.
\end{equation}
The expression of interest is thus
\begin{equation}\label{E:Pole2}
h^{(m)}(z) = \frac1{z^m}\,\sum\limits_{n=1}^{\infty}  \frac{S_{m\,n}}{z^n} \sum\limits_{k=1}^{N_k}\, z_k^n{\nu}_k = 0\ ,
\end{equation}
which leads to the following matrix condition for the coefficients:
\begin{equation}\label{E:Pole3}
\mathbf{G}\mathbf{X}\,{\mathbf{\nu}} = 0
\end{equation}
with ${\mathbf{\nu}} \eq ({\nu}_1,\,{\nu}_2,\,{\nu}_3,\,\cdots,\,{\nu}_{N_k})^T$. The solvability of this matrix equation directly yields a theorem expressing the uniqueness of higher-order pole representations: 
\begin{theorem}\label{T:HigherPoleUniq}
For finite $m > 0$, if $h^{(m)}(z)$ has the form specified by (\ref{E:Pole1}), where $N_k$ is finite, $0 < |z_{k}| < 1$ and $z_{k'} \neq z_k$ for $k' \neq k$, then $h^{(m)}(z) = 0$ for all $|z| \geq 1$ if and only if ${\nu}_k = 0$ for all $k$. 
\end{theorem}
\noindent
Here as $m \longrightarrow \infty$ the theorem no longer holds since $(z - z_k)^{-m}\longrightarrow 0$ for all $|z| \geq 1$.  Also Theorem~\ref{T:HigherPoleUniq} pertains only to poles of the same order.  For very small $N_k$ it is easy to show that expansions of mixed orders of poles are linearly independent by direct algebraic means.

\section{Uniqueness Results for $\mathbb{R}^2$
}\label{S:Real2}

  Comparing (\ref{E:PoleExpansion}), (\ref{E:Mk1}) and (\ref{E:Mk2}) with (\ref{E:SimplePoles}) shows that
\begin{equation}\notag
 {\text{Re}}\,\{h(z)\} =  \sum\limits_{k=1}^{N_k}\ {\vec{D}}_k\mathbf{\cdot}\,{\mathbf{\nabla}}\,\,\ln\ (|{\vec{X}} - {\vec{X}}_k|^{-1})
\end{equation}
where ${\vec{D}}_k = (-{\alpha}_k,\,-{\beta}_k)^T$, which immediately shows that dipole expansions in \R{2} are unique by Theorems~\ref{T:CtoRmap} and \ref{T:SmpPoleUniq}.  In a like fashion, it is clear that Theorem~\ref{T:HigherPoleUniq} in conjunction with Theorem~\ref{T:CtoRmap} (with a passive transformation coefficient used as needed) also implies that higher order \R{2} multipole expansions of order $n$ are unique if they can be fully represented by two basis functions at each location. As noted in Section~\ref{S:intro}, in addition to the $n = 1$ case, this program can be carried out for $n = 2$ and $3$, but for $n > 3$ there are more than two independent multipole basis functions for each ${\vec{X}}_k$ and so this correspondence cannot be one-to-one.  

  First, it is useful to build on the multipole notation introduced in conjunction with (\ref{E:Mk1}) and (\ref{E:Mk2}) by letting an $i$, $i'$, $j$ or $j'$ subscript preceeded by a comma denote $\partial/\partial x$ when the subscript in question takes on the value $1$ and to denote $\partial/\partial y$ when it takes on the value $2$ (in what follows it will be assumed that $i$, $i'$, $j$ and $j'$ always take on the values $1$ to $N$ in \R{N} and that similar partials are implied for \R{N}, but only the \R{2} case will be considered in this section).  Multipole basis functions of any order can then be defined by expressions of the form
\begin{equation}\notag
{\Mop}_{i\,i'\,j\,\ldots\,j'}^{{}_{[k]}} \equiv {\Mop}_{i\,,i'\,,j\,,\ldots,\,j'}^{{}_{[k]}}
\end{equation}
where (as noted before) the number of subscripts corresponds to the order of the multipole.
For $n = 2$ there are only two independent basis functions, say ${\Mop}_{1\,1}^{{}_{[k]}}$ and ${\Mop}_{1\,2}^{{}_{[k]}}$ since obviously ${\Mop}_{2\,1}^{{}_{[k]}} = {\Mop}_{1\,2}^{{}_{[k]}}$ and from Laplace's equation ${\Mop}_{2\,2}^{{}_{[k]}} = - {\Mop}_{1\,1}^{{}_{[k]}}$.  Likewise for $n = 3$, if ${\Mop}_{1\,1\,1}^{{}_{[k]}}$ and ${\Mop}_{2\,2\,2}^{{}_{[k]}}$ are selected as the two independent basis functions, then ${\Mop}_{1\,1\,2}^{{}_{[k]}} = -{\Mop}_{2\,2\,2}^{{}_{[k]}}$, ${\Mop}_{2\,2\,1}^{{}_{[k]}} = -{\Mop}_{1\,1\,1}^{{}_{[k]}}$ and analogous relationships hold for permuted indices.  For $n = 4$ it is easy to check that all of the other basis functions can be obtained from ${\Mop}_{1\,1\,1\,1}^{{}_{[k]}}$, ${\Mop}_{2\,2\,2\,2}^{{}_{[k]}}$ and ${\Mop}_{1\,1\,1\,2}^{{}_{[k]}}$.  This shows the set of quadrupole basis functions for some given source point cannot be larger than the indicated set, but it does not show that the indicated sets are indeed independent.  Showing the actual independence of the indicated basis functions at a common source point follows from direct algebraic manipulation and involves first calculating them explicitly and then multiplying through by the common denominator, but this straightforward algebraic manipulation will not be done here.

\newcommand{\Aop}{\mathbb{A}}
\newcommand{\Bop}{\mathbb{B}}

   For $n < 4$, since there are only two independent basis multipole basis functions of order $n$ at each source location, it is useful to introduce a more compact notation for these cases so that the resulting potential can be written more efficiently in terms of the above independent basis multipole functions
\begin{equation}\label{E:ABseries}
 H^{(n)}(\vec{X}) \equiv \sum\limits_{k=1}^{N_k} \left(a^{(n)}_k{\Aop}^{(n)}_k + b^{(n)}_k{\Bop}^{(n)}_k\right)
\end{equation}
where, for $n = 1$
\begin{equation}\label{E:AB1def}
 {\Aop}^{(1)}_k \equiv {\Mop}_{1}^{{}_{[k]}}\ \ \text{and}\ \  {\Bop}^{(1)}_k \equiv {\Mop}_{2}^{{}_{[k]}}\,,
\end{equation}
for $n = 2$
\begin{equation}\label{E:AB2def}
 {\Aop}^{(2)}_k \equiv {\Mop}_{1\,1}^{{}_{[k]}}\ \ \text{and}\ \  {\Bop}^{(2)}_k \equiv {\Mop}_{1\,2}^{{}_{[k]}}\,,
\end{equation}
and for $n = 3$
\begin{equation}\label{E:AB3def}
 {\Aop}^{(3)}_k \equiv {\Mop}_{1\,1\,1}^{{}_{[k]}}\ \ \text{and}\ \  {\Bop}^{(3)}_k \equiv {\Mop}_{2\,2\,2}^{{}_{[k]}}\ .
\end{equation}
Also in (\ref{E:ABseries}) the coefficients $a^{(n)}_k$ and $b^{(n)}_k \in \mRR$.  Subsequently the superscript indicating the order is often omitted from $a^{(n)}_k$ and $b^{(n)}_k$. 

  Explicit expressions for ${\Aop}^{(n)}_k$ and ${\Bop}^{(1)}_k$ can be easily written:
\begin{subequations}\label{E:Avals}
\begin{align}
{\Aop}^{(1)}_k &= \frac{{\partial}\ }{\partial x} \ln\ (|{\vec{X}} - {\vec{X}}_k|^{-1}) = -\,\frac{x-x_k}{|{\vec{X}} - {\vec{X}}_k|^{2}}\ \,, &{\Bop}^{(1)}_k &= -\,\frac{x-x_k}{|{\vec{X}} - {\vec{X}}_k|^{2}}\label{E:Av1}\\
{\Aop}^{(2)}_k  &= \frac{(x-x_k)^2 - (y - y_k)^2}{|{\vec{X}} - {\vec{X}}_k|^{4}}\ \,,
&{\Bop}^{(2)}_k  &= 2\frac{(x-x_k)(y - y_k)}{|{\vec{X}} - {\vec{X}}_k|^{4}}\label{E:Av2}\\
{\Aop}^{(3)}_k  &= \frac{2(x-x_k)}{|{\vec{X}} - {\vec{X}}_k|^{6}}[3(y - y_k)^2 - (x-x_k)^2] \ \ \ \ {\text{and}}  
&{\Bop}^{(3)}_k  &= \frac{2(y-y_k)}{|{\vec{X}} - {\vec{X}}_k|^{6}}[3(x - x_k)^2 - (y-y_k)^2].\label{E:Av3}
\end{align}
\end{subequations}
 The higher order complex poles corresponding to the expressions in (\ref{E:Avals}) for $n = 2$ and $n = 3$ can easily be found, as in (\ref{E:PoleExpansion}), by multiplying the numerator and denominator of ${\mu}_k/(z - z_k)^n$ by $(z^{*} - z^{*}_k)^n$:
\begin{subequations}\label{E:Pvals}
\begin{align} 
{\text{Re}}\,\,\Big{\{}\frac{{\mu}_k}{(z - z_k)^2}\Big{\}} &= \frac{{\alpha}_k[(x - x_k)^2 - (y -y_k)^2]}{|z - z_k|^4} +  \frac{2{\beta}_k(x - x_k)(y -y_k)}{|z - z_k|^4}\\
{\text{Re}}\,\,\Big{\{}\frac{{\mu}_k}{(z - z_k)^6}\Big{\}} &= \frac{{\alpha}_k(x - x_k)[(x - x_k)^2 - 3(y -y_k)^2]}{|z - z_k|^6} + \frac{{\beta}_k(y -y_k)[3(x - x_k)^2 - (y -y_k)^2]}{|z - z_k|^6}
\end{align}
\end{subequations}
Since $|{\vec{X}} - {\vec{X}}_k| = |z - z_k|$, from (\ref{E:Avals}) it thus follows that (\ref{E:Pvals}) can be immediately be rewritten as
\begin{subequations}\label{E:PAmap}
\begin{align} 
{\text{Re}}\,\,\Big{\{}\frac{{\mu}_k}{(z - z_k)^2}\Big{\}} &= \ \  {\alpha}_k{\Aop}^{(2)}_k\  + {\beta}_k{\Bop}^{(2)}_k\label{E:PAmap1}\\
{\text{Re}}\,\,\Big{\{}\frac{{\mu}_k}{(z - z_k)^3}\Big{\}} &= - \frac{{\alpha}_k}2{\Aop}^{(3)}_k + \frac{{\beta}_k}2{\Bop}^{(3)}_k\label{E:PAmap12}
\end{align}
\end{subequations}
Clearly equations~(\ref{E:PAmap}) imply an invertible passive coefficient mapping, so that ${\text{Re}}\,\{h^{n)}(z)\} \iff H^{(n)}({\vec{X}})$ holds.  Thus Theorem~\ref{T:HigherPoleUniq} implies that multipole expansions of order three or less in \R{2} are unique and these results can be summarized in the following formally as:
\begin{theorem}\label{T:R2HigherPoleUniq}
In \R{2}, for finite $0 < n < 4$, if $H^{(n)}(\vec{X})$ has the form specified by (\ref{E:ABseries})
 where $N_k$ is finite, $0 < |{\vec{X}}_{k}| < 1$ and ${\vec{X}}_{k'} \neq {\vec{X}}_k$ for $k' \neq k$, then $H^{(n)}(\vec{X}) = 0$ for all $|\vec{X}| \geq 1$ if and only if ${a}_k = 0$ and $b_k = 0$ for all $k$. 
\end{theorem}
\noindent
or less formally as
\begin{theorem}\label{T:HigherPoleUniq2}
A finite expansion of \R{2} points masses, point dipoles or point quadrupoles is unique.
\end{theorem}

   Since there are more than two multipole basis functions for each source location of an \R{2} multipole of order four or higher, no simple correspondence exists between a given single higher order complex pole and the corresponding multipoles at the same location.  Given that all of the uniqueness results considered so far rest on such a correspondence, it is clear that most of the readily obtainable results in \R{2} have been found and multipoles of order $n > 3$ will thus not be considered.  The next issue to be addressed is extending these \R{2} results to \R{3} or even \R{N}, for $N > 3$.

\section{Uniqueness Results for $\mathbb{R}^3$}\label{S:Real3}

  This section addresses the question of point mass uniqueness by showing that uniqueness of point mass distributions in \R{2} has direct consequences in \R{3}.  Towards that end the following simple lemma will prove to be useful:
\begin{lemma}\label{T:RNLemma}
In \R{N}, for $N > 1$, an array of $N_k$ distinct points forms at most $2\,N_k \times (N_k - 1)$ unique directions when all possible lines containing two or more array points are considered; moreover, it is possible to rotate coordinates so that all of the points are distinct when projected into the $N-1$ dimensional hyperplane orthogonal to some preferred coordinate axis.
\end{lemma}
\begin{proof}
Since two points determine a straight line, the maximin number of independent lines occurs when no three or more points are collinear; i.e., chose any of the $N_k$ points and a second distinct point.  Factoring in the fact that any line determines two possible directions yields $2\,N_k \times (N_k - 1)$ for a maximum number of possible directions. (Besides the fact that many points may be collinear, many of the lines formed may be parallel, so the actual number might be much smaller than $2\,N_k \times (N_k - 1)$.)  This shows the first part of the lemma.

  Next consider the second part of the lemma.
  Given a choice of coordinate origin, there is a continuous choice of preferred axis directions.  Specifically, for a preferred coordinate axis ($N$'th axis) choose a direction that does not coincide with any of the finite possible directions along which two or more points line up.  Since no two points line up, their projection into the orthogonal hyperplane of the preferred direction is distinct and the lemma follows.  
\end{proof}
This lemma leads to the following lemma:
\begin{lemma}\label{T:PMLemma}
If a point mass uniqueness counter-example exists in \R{3}, then one exists in \R{2}. 
\end{lemma}
\begin{proof}
If a counter example exists in \R{3} then from (\ref{E:PM3}) it can be assumed, without loss of generality, that an $N_k$ exists with $m_k \neq 0$ for all $N_k \geq k \geq 1$ such that 
\begin{equation}\label{E:PMU}
V(\vec{X}) = \sum\limits_{k=1}^{N_k}\ \frac{m_k\ \ }{|{\vec{X}} - {\vec{X}}_k|} = 0
\end{equation}
for all $|\vec{X}| \geq 1$ (where ${\vec{X}}_{k'} \neq {\vec{X}}_{k}$ for all $k' \neq k$).  As before, without loss of generality, it is assumed that $\sum_{k=1}^{N_k} m_k = 0$ and ${\vec{X}}_{k}\neq 0$ for all $k$.  Thus, for later convenience a set of point masses $\sum_{k=1}^{N_k} m_k = 0$ at the origin can be explicitly added to (\ref{E:PMU}).

 Let ${\vec{X}} \equiv (x,\,y,\,z)^T$ and ${\vec{X}}_k \equiv (x_k,\,y_k,\,z_k)^T$ where there is no danger here of confusing $z$ and $z_k$ with the complex variables introduced earlier.  
Assume that the coordinates have been chosen in accord with Lemma~\ref{T:RNLemma} where the preferred direction has been taken to be along the $z-$axis so that in the $x-y$ plane the $N_k$ points are all distinct, so that $(x_k,\,y_k,\,0)^T \neq (x_{k'},\,y_{k'},\,0)^T$ for all $k' \neq k$.  Clearly a linear superpositions of potentials along the $z-$axis also obey the criteria of (\ref{E:PMU}) so that 
\begin{equation}\label{E:Zint}
\int\limits_{z=-L}^{L} \sum\limits_{k=1}^{N_k}\ \left(\frac{m_k\ \ }{|{\vec{X}} - {\vec{X}}_k|} - \frac{\ m_k\ \ }{|{\vec{X}}|\ }\right)\,d\,z\ =\ 0
\end{equation}
Observe that 
\begin{equation}\notag
\lim_{\,\ \ L \to \infty}\, \int\limits_{z=-L}^{\ L}\frac{\ \ d\,z}{\sqrt{(x - x_k)^2 + (y - y_k)^2 + (z - z_k)^2}}\ =\ \lim_{\,\ \ L \to \infty}\, 2\!\!\int\limits_{z=0}^{\ L}\frac{\ \ d\,z}{\sqrt{(x - x_k)^2 + (y - y_k)^2 + z^2}}
\end{equation}
and that since 
\begin{equation}\notag
\int\limits_{z=0}^{\ L}\frac{\ \ d\,z}{\sqrt{a^2 + z^2}}\ d\,z = \ln\,\left(L + \sqrt{a^2 + L^2}\,\right) - \ln\, a
\end{equation}
it follows that
\begin{equation}\notag
\int\limits_{z=-L}^{\ L}\!\!\left(\frac{1}{\sqrt{a^2 + z^2}} - \frac{1}{\sqrt{b^2 + z^2}}\right)\,d\,z\  =\ 2\,\ln\left(\frac{1 + \sqrt{(a/L)^2 + 1}}{1 + \sqrt{(b/L)^2 + 1}}\right) -  2\,\ln\,\left(\frac{a}{b}\right),
\end{equation}
where $a \equiv \sqrt{(x - x_k)^2 + (y - y_k)^2}\,$ and $b \equiv \sqrt{x^2 + y^2}$.  It thus follows that
\begin{equation}\label{E:Zinf}
\int\limits_{z=-\infty}^{\ \ \infty}\!\left(\frac1{|{\vec{X}} - {\vec{X}}_k|} - \frac{1}{|{\vec{X}}|\ }\right)\,d\,z\ =  2{\Psi}_k(\vec{X})
\end{equation}
and thus that 
\begin{equation}\label{E:R2counter}
\sum\limits_{k=1}^{N_k}\, m_k{\Psi}_k(\vec{X}) = 0\ \ \text{with}\ \ m_k \neq 0  \ \text{for all}\ k.
\end{equation}
\end{proof}
\noindent
A comparison of (\ref{E:R2counter}) and (\ref{E:PMPsi}) allows one to one to conclude from Lemma~\ref{T:PMLemma} and Theorem~\ref{T:HigherPoleUniq2} that
\begin{theorem}\label{T:HigherPoleUniq3}
A finite expansion of \R{3} points masses is unique, where the usual conditions are assumed to apply.
\end{theorem}

  Next consider the question of uniqueness for the \R{3} dipole expansion given by (\ref{E:DP3}):
\begin{equation}\label{E:DPR3}
W(\vec{X}) \equiv \sum\limits_{k=1}^{N_k}\ {\vec{D}}_k\mathbf{\cdot}\,{\mathbf{\nabla}}\,({|{\vec{X}} - {\vec{X}}_k|^{-1}})\ .
\end{equation}
It is possible to state and prove a dipole analog of Lemma~\ref{T:PMLemma}:
\begin{lemma}\label{T:DPLemma}
If a point dipole uniqueness counter-example exists in \R{3} then one exists in \R{2}. 
\end{lemma}
\begin{proof}
Suppose, as in Lemma~\ref{T:PMLemma}, to the contrary that an \R{3} counterexample exists and thus that for some $N_k$ distinct ${\vec{X}}_k$, that for some $W$ given by the right hand side of (\ref{E:DPR3}) $W \equiv 0$ for some $|{\vec{D}}_k| \neq 0$, for all $k$.  From Lemma \ref{T:RNLemma} there are only a finite number of directions parallel to the lines determined by two or more of the ${\vec{X}}_k$'s. To this finite collection of directions add the vectors $\{{\vec{D}}_k\}_{k=1}^{N_k}$ and then select a preferred $z$ coordinate direction that is different from all of these directions and denote this direction by $\hat{k}$.  Since $|\hat{k}\cdot{\vec{D}}_k| < |{\vec{D}}_k|$ for all $k$ and, by construction, not only does $(x_k,\,y_k,\,0)^T \neq (x_{k'},\,y_{k'},\,0)^T$ for all $k' \neq k$ hold, as in the proof of Lemma~\ref{T:PMLemma}, but the projection of ${\vec{D}}_k$ on the $x-y$ plane is nonzero.  Let ${\vec{D}}^{\{2\}}_k$ denote this projection.  Further, as before, the case $\sum_{k=1}^{N_k} {\vec{D}}_k \neq 0$ presents no real difficulties and $(x_k,\,y_k,\,0)^T \neq 0$ can also be assumed.  Then taking the integral along the $z-$axis as in (\ref{E:Zinf}) gives [after adding an analogous term at the origin to the one added to (\ref{E:Zint})]
\begin{subequations}\label{E:DZinf}
\begin{align}
\int\limits_{z=-\infty}^{\ \ \infty}\!\frac{\partial\ }{\partial\,x}\left(\frac1{|{\vec{X}} - {\vec{X}}_k|} - \frac{1 }{|{\vec{X}}|\ }\right)\,d\,z\ &=  2\frac{\partial\ }{\partial\,x}{\Psi}_k(\vec{X})\label{E:DZinfX}\\
\int\limits_{z=-\infty}^{\ \ \infty}\!\frac{\partial\ }{\partial\,y}\left(\frac1{|{\vec{X}} - {\vec{X}}_k|} - \frac{1 }{|{\vec{X}}|\ }\right)\,d\,z\ &=  2\frac{\partial\ }{\partial\,y}{\Psi}_k(\vec{X})\label{E:DZinfY}\\
\int\limits_{z=-\infty}^{\ \ \infty}\!\frac{\partial\ }{\partial\,z}\left(\frac1{|{\vec{X}} - {\vec{X}}_k|} - \frac{1 }{|{\vec{X}}|\ }\right)\,d\,z\ &=  0
\end{align}
\end{subequations}
Using (\ref{E:DZinf}) in (\ref{E:DPR3}) allows the \R{3} counterexample to be restated as
\begin{equation}\label{E:Wint}
\int\limits_{z=-\infty}^{\ \ \ \infty}\negthickspace\negthickspace W(\vec{X}) \ d\,z\ =\  2\sum\limits_{k=1}^{N_k}\, {\vec{D}}^{\{2\}}_k{\mathbf{\cdot}}{\nabla}{\Psi}_k(x,\,y)\ =\ 0
\end{equation}
where $|{\vec{D}}^{\{2\}}_k|\, \neq 0$ for all $k$.
\end{proof}
Since Theorem~\ref{T:HigherPoleUniq2} shows that a counterexample of the form specified by (\ref{E:Wint}) cannot exist, the following theorem is an immediate consequence of Lemma~\ref{T:DPLemma}:
\begin{theorem}\label{T:dipoleUniq3}
A finite expansion of \R{3} point dipoles is unique.
\end{theorem}

 Clearly, one would expect the results of Lemma~\ref{T:PMLemma} and thus Theorems~\ref{T:HigherPoleUniq3} and \ref{T:dipoleUniq3} to generalize to \R{N} for $N > 3$, with only the technical difficulty of explicitly obtaining the $N$ dimensional integral analogs of (\ref{E:Zinf}) and (\ref{E:DZinf}) to stand in the way; however, two other remaining issues of far greater practical importance are less clear:  (1) Proving quadrupole and other multipole uniqueness results in \R{3}.  (2) Proving expansions of mixed types of point sources are unique.
It is unclear how the first issue should be approached, but an approach to the second issue can be based on the theorems of Section~\ref{S:Umaps}.   As before, results in the complex plane will be the starting point.

\section{Uniqueness for Point Sources of Mixed Type}\label{S:Mixed}

This section deals with uniqueness results for expansions of mixed type.  Only the complex setting will be considered here even though results for mixed types in \C\ can be directly extended to \R{2} and \R{3} mixed type results.  There is a pragmatic reason for considering only the complex setting: Uniqueness results for mixed types of expansions are very difficult to prove and, in the end, given that there is only a limited amount of success here in studying the complex case, consideration of the real case is irrelevant.

  Thus two particular kinds of mixed type analytic expansions will be of primary interest here: (1) Expansions consisting of both simple pole terms and second order pole terms. (2) Expansions consisting of simple poles and logarithmic basis functions.   
 
First, observe that as far as uniqueness results for expansions of mixed type are concerned, it only necessary to consider the general case where there are (possibly) a like number of either simple poles and logarithmic basis functions that are located at the same point or, in a like manner, to consider only the case consisting of simple poles and second order poles that are located at the same point.  Thus, without loss of generality, for the simple pole and second order pole case consider the following expression for poles of mixed type:
\begin{equation}\label{E:MixedC}
h^{(1,\,2)}(z) \equiv \sum\limits_{k=1}^{N_k} \frac{{\mu}_k}{(z - z_k)}\ + \sum\limits_{k=1}^{N_k} \frac{{\nu}_k}{(z - z_k)^{2}}
\end{equation}
For $m = 2$, (\ref{E:BinCoef}) gives
\begin{equation}
\frac{1}{(z - z_k)^{2}} = \frac1{z^2} \sum\limits_{n=0}^{\infty} (n + 1) \frac{z_k^n}{z^n}\ ,
\end{equation}
which can be combined with (\ref{E:geoseries}) to yield
\begin{equation}\label{E:MixedC1}
h^{(1,\,2)}(z) = \frac1z\sum\limits_{k=1}^{N_k} {{\mu}_k}\sum\limits_{n=0}^{\infty} \frac{z_k^n}{z^{n}}\ + \frac1{z^2}\sum\limits_{k=1}^{N_k}{{\nu}_k} \sum\limits_{n=0}^{\infty} (n + 1)\frac{z_k^n}{z^n}\ .
\end{equation}
The RHS of (\ref{E:MixedC1}) can be rewritten as
\begin{equation}\label{E:MixedC2}
z{\cdot}h^{(1,\,2)}(z) = \sum\limits_{k=1}^{N_k} {{\mu}_k} + \sum\limits_{k=1}^{N_k} {{\mu}_k}\sum\limits_{n=1}^{\infty} \frac{z_k^n}{z^{n}}\ + \sum\limits_{k=1}^{N_k}{{\nu}_k} \sum\limits_{n=1}^{\infty} \frac{nz_k^{n-1}}{z^n}\ .
\end{equation}
Setting the successive powers of $z^{-n}$ individually to zero on the RHS of (\ref{E:MixedC2}) (and simply ignoring the condition $\sum_{k=1}^{N_k} {{\mu}_k} = 0$) gives the following set of equations that must hold when $z{\cdot}h^{(1,\,2)}(z) = 0$:
\begin{equation}\label{E:Gnew}
    \begin{pmatrix}
  z_1    & z_2    & z_3    & \cdots & z_{N_k}   &  1    & 1     & 1     & \cdots & 1\\
\\
  z_1^2  & z_2^2  & z_3^2  & \cdots & z_{N_k}^2 & 2z_1  & 2z_2  & 2z_3  & \cdots & 2z_{N_k} \\
\\
  z_1^3  & z_2^3  & z_3^3  & \cdots & z_{N_k}^3 & 3z_1^2& 3z_2^2& 3z_3^2& \cdots & 3z_{N_k}^2\\
     \vdots & \vdots & \vdots &        & \vdots & \vdots& \vdots& \vdots&        & \vdots\\
  z_1^{{}^{N_k}}  & z_2^{{}^{N_k}}  & z_3^{{}^{N_k}}  & \cdots &z_{N_k}^{{}^{N_k}} & \mSmall{N_k}\!z_1^{{}^{N_k-1}} & \mSmall{N_k}\!z_2^{{}^{N_k-1}}  &\mSmall{N_k}\!z_3^{{}^{N_k-1}}  &  \cdots & \mSmall{N_k}\!z_{N_k}^{{}^{N_k-1}} \\
 \\
  z_1^{{}^{N_k+1}} \mspace{-14mu} & \mspace{-5mu}z_2^{{}^{N_k+1}}  \mspace{-14mu}  & \mspace{-5mu}z_3^{{}^{N_k+1}}  \mspace{-12mu}  &  \mspace{-18mu}\cdots  \mspace{-10mu}  &  \mspace{-14mu} z_{N_k}^{{}^{N_k+1}}\mspace{-5mu} & \mspace{-11mu}\mSmall{(N_k\!+\!1)}\!z_1^{{}^{N_k}}\mspace{-5mu} &\mspace{-11mu} \mSmall{(N_k\!+\!1)}\!\!z_2^{{}^{N_k}} \mspace{-5mu}  & \mspace{-11mu}\mSmall{(N_k\!+\!1)}\!z_3^{{}^{N_k}} \mspace{-5mu}   & \mspace{-18mu} \cdots \mspace{-12mu}& \mspace{-11mu}\mSmall{(N_k\!+\!1)}\!z_{N_k}^{{}^{N_k}}\mspace{-6mu}  \\
     \vdots & \vdots & \vdots &        & \vdots & \vdots& \vdots& \vdots&        & \vdots\\
        \end{pmatrix}
  \begin{pmatrix}
 \mu_1 \\
 \mu_2 \\
 \mu_3 \\
 \vdots\\
 \,\,\ \mu_{{}_{N_k}}\\
 \nu_1 \\
 \nu_2 \\
 \nu_3 \\
 \vdots \\
 \,\,\ \nu_{{}_{N_k}}
\end{pmatrix}
 = 0\ .
\end{equation}
% \mspace{-1mu} \negmedspace \negthickspace 
The first $2N_k$ rows of this matrix equation set can immediately be rewritten in the following block matrix form:
\begin{equation}\label{E:GPP}
\left(
\begin{array}{r|l}
 {\mathbf{G}}_{{\vphantom{[}}_{\vphantom{[}}}\mathbf{X}\ \ \ \ \ & \ \ \ \ \mathbf{N}{\mathbf{G}}_{{\vphantom{[}}}\\ \hline
 {\vphantom{[}}^{{\vphantom{[}}^{\vphantom{[}}}{\mathbf{G}}_{{\vphantom{[}}_{\vphantom{[}}}{\mathbf{X}}^{{}^{N_k+1}} & \vphantom{[}({\mathbf{N}+N_k\mathbf{I}){\mathbf{G}}_{\vphantom{[}}}{\mathbf{X}}^{{}^{N_k}}
\end{array}  % Here the subscript and superscript {\vphantom{[}} usage opens up vertical 
%              spacing.
\right)\!
\left(
\begin{array}{c}
 {\mathbf{\mu}}_{ {{{\vphantom{R}}_{\vphantom{[}}}}_{\vphantom{[}} }\\ \hline
 \mathbf{\nu}_{\vphantom{R}}^{\vphantom{[}}
\end{array}
\right)\,
 = 0\, .
\end{equation}
Here $\mathbf{I}$ is the $N_k \times N_k$ identity matrix and the other matrices in (\ref{E:GPP}) have been previously defined.  (\ref{E:GPP}) is equivalent to the two coupled equations sets
\begin{align}
   \mathbf{G}\mathbf{X}\,\mu + \mathbf{N}\mathbf{G}\nu\  &=\  0\\
    \mathbf{G}\mathbf{X}{{}^{N_k+1}}\mu + (\mathbf{N} + N_k\mathbf{I})\mathbf{G}{\mathbf{X}}^{{}^{N_k}}\!\!\nu\  &=\  0
\end{align}
Here the first of these equations uniquely determines $\mathbf{\mu}$ in terms of $\mathbf{\nu}$:
\begin{equation}\label{E:muORnu}
 \mu = -{\mathbf{X}}^{-1}{\mathbf{G}}^{-1}\mathbf{N}\mathbf{G}\,\nu\ .
\end{equation}
When this is substituted into the second matrix equation and the result is rearranged slightly the following matrix equation for $\mathbf{\nu}$ results 
\begin{equation}\label{E:nuEqn}  
   [\mathbf{G}\mathbf{X}{{}^{N_k}}( N_k\mathbf{I} - {\mathbf{G}}^{-1}\mathbf{N}\mathbf{G})\, + \, \mathbf{N}\mathbf{G}{\mathbf{X}}^{{}^{N_k}}]\,\nu\  =\  0\ ,
\end{equation}
which can be multiplied from the left by ${\mathbf{X}}^{-N_k}{\mathbf{G}}^{-1}$ to give
\begin{equation}\label{E:nuEqn2}
[( N_k\mathbf{I} - {\mathbf{G}}^{-1}\mathbf{N}\mathbf{G})+ {\mathbf{X}}^{{}^{-N_k}}{\mathbf{G}}^{-1}\mathbf{N}\mathbf{G}{\mathbf{X}}^{{}^{N_k}}] \,\nu\  =\  0\,.
\end{equation}
Here (\ref{E:nuEqn}) immediately implies that if the matrix $[( N_k\mathbf{I} - {\mathbf{G}}^{-1}\mathbf{N}\mathbf{G})+ {\mathbf{X}}^{{}^{-N_k}}{\mathbf{G}}^{-1}\mathbf{N}\mathbf{G}{\mathbf{X}}^{{}^{N_k}}]$ is invertible then $\nu = 0$ must hold.  In turn, this matrix is invertible if the following matrix is invertible:
\begin{equation}\label{E:Cmatrix}
\mathbf{C} \eq [N_k\mathbf{I} - \mathbf{N} + \mathbf{G}{\mathbf{X}}^{-N_k}{\mathbf{G}}^{-1}\mathbf{N}\mathbf{G}{\mathbf{X}}^{{}^{N_k}}{\mathbf{G}}^{-1}]\,,
\end{equation}
which can be rewritten as
\begin{equation}\label{E:Cmatrix2}
\mathbf{C} \eq  N_k\mathbf{I} - \mathbf{N} + {\mathbf{U}}^{-1}\mathbf{N}{\mathbf{U}}, \ \text{where}\ \ \mathbf{U} \eq \mathbf{G}{\mathbf{X}}^{{}^{N_k}}{\mathbf{G}}^{-1}\,.
\end{equation}
Showing that $\mathbf{C}$ is invertible is not as easy at it might first appear.  Thus, for example, first observe that $\mathbf{C}$ can be rewritten as $\mathbf{C} = \mathbf{A} + \mathbf{B}$, where $\mathbf{A} \eq  N_k\mathbf{I} - \mathbf{N} $ and $\mathbf{B} \eq {\mathbf{U}}^{-1}\mathbf{N}{\mathbf{U}}$.  Obviously $\mathbf{A}$ is positive definite; furthermore, it is clear that the sum of two positive definite matrices is positive definite, so if it can be shown that $\mathbf{B}$ is positive definite, then $\mathbf{C}$ will be positive definite and thus invertible.  Next observe that it is trivial to find the eigenvectors of $\mathbf{B}$ and that the corresponding eigenvalues are given by the diagonal elements of $\mathbf{N}$.  As one can quickly convince him or herself, this, however, does not imply that $\mathbf{B}$ is positive definite, since, among other things, the eigenvectors are not orthogonal to each other (clearly $\mathbf{B}$ is not normal since $\mathbf{B}\mathbf{B}^T \neq \mathbf{B}^T\mathbf{B}$ and normality is assumed in relevant theorems of interest).  Given that the author's attempts at proving that $\mathbf{C}$ is invertible have not met with success, the issue is open.   

 Here the question naturally arises as to what the analog of the $\mathbf{C}$ matrix would have been if the set of $n$ rows starting at row $mn + 1$ had been taken instead of at row $n + 1$. In this case the system of equations that result are 
\begin{equation}\label{E:nuEqn4}
{\mathbf{C}}_m\,\nu\  =\  0\ ,   
\end{equation}
where:
\begin{equation}\label{E:Cmatrix3}
{\mathbf{C}}_m \eq  mN_k\mathbf{I} - {\mathbf{G}}^{-1}\mathbf{N}{\mathbf{G}}^{-1} + {\mathbf{X}}^{{}^{-mN_k}}{\mathbf{G}}^{-1}\mathbf{N}{\mathbf{G}}{\mathbf{X}}^{{}^{mN_k}}\,.
\end{equation}

  Here it is also natural to also raise the question about what the analog of the $\mathbf{C}$ is when logarithmic basis functions are included. 
  Thus consider the conditions that must hold if $\varphi(z) = h^{(1,\,2,\,3)}(z) = 0$, where
\begin{equation}\label{E:h123}
h^{(1,\,2,\,3)}(z) \eq \sum\limits_{k=1}^{N_k}\Blbrac\rho_k\psi_k(z) + \frac{{\mu}_k}{(z - z_k)}\ +  \frac{{\nu}_k}{(z - z_k)^{2}}\Brbrac\ .
\end{equation}
Since 
\begin{equation}
\frac{d\,\psi_k(z)}{d\,z\ \ \ } = \frac1z - \frac1{z - z_k} = -\sum\limits_{n=1}^{\infty}\frac{z_k^{n}}{z^{n+1}} = -\sum\limits_{n=2}^{\infty}\frac{z_k^{n-1}}{z^{n}}
\end{equation} 
it is obvious that the correct power series expansion for $\psi_k(z)$ is
\begin{equation}
\psi_k(z) = \sum\limits_{n=1}^{\infty}\frac1n\frac{z_k^{n}}{z^{n}}
\end{equation}
Substituting this expansion along with
\begin{subequations}\label{E:abc}
\begin{align}
\frac1{z - z_k} &= \sum\limits_{n=0}^{\infty}\frac{z_k^n}{z^{n+1}} = \sum\limits_{n=1}^{\infty}\frac{z_k^n-1}{z^{n}} \label{E:abc1}\\
\intertext{and}
\frac1{(z - z_k)^2} &= \sum\limits_{n=1}^{\infty}\frac{nz_k^{n-1}}{z^{n+1}} = \sum\limits_{n=2}^{\infty}\frac{(n - 1)z_k^{n-2}}{z^{n}}\label{E:abc2}
\end{align}
\end{subequations}
yields
\begin{equation}
h^{(1,\,2,\,3)}(z)  = \sum\limits_{k=1}^{N_k}\Blbrac \sum\limits_{n=1}^{\infty}\frac{\rho_k}n\frac{z_k^{n}}{z^{n}} + \sum\limits_{n=1}^{\infty}\frac{z_k^{n-1}}{z^{n}}{\mu_k} + 
\sum\limits_{n=2}^{\infty}\frac{(n - 1)z_k^{n-2}}{z^{n}}{\nu_k}\Brbrac 
\end{equation}
Breaking out the $n = 1$ terms separately and reindexing yields:
\begin{equation}
h^{(1,\,2,\,3)}(z)  = \frac1{z}\,\sum\limits_{k=1}^{N_k}\left( \rho_kz_k + {\mu_k} \right) + \frac1z\sum\limits_{n=1}^{\infty}\frac1{z^n}\sum\limits_{k=1}^{N_k}\Blbrac \frac{\rho_kz_k^{n+1}}{(n + 1)} + {z_k^{n}}{\mu_k} +  nz_k^{n-1}{\nu_k}\Brbrac 
\end{equation}
Introducing ${{\rho}'_k} = z_k\rho_k$ and ${\nu}'_k = {\mu_k}/z_k$ yields
\begin{equation}\label{E:rhoprime}
h^{(1,\,2,\,3)}(z)  = \frac1z\,\sum\limits_{k=1}^{N_k}\left(\rho_kz_k + {\mu_k} \right) + \frac1z\sum\limits_{n=1}^{\infty}\frac1{z^n}\sum\limits_{k=1}^{N_k}\Blbrac \frac{{\rho}'_k}{(n + 1)} + {\mu_k} +  n{{\nu}'_k}\Brbrac z_k^n\ .
\end{equation}
Ignoring the first term on the RHS of (\ref{E:rhoprime}) and setting the other powers of $1/z$ to zero yields an equation set whose first $3n$ rows can be written in block matrix form as:
\begin{equation}\label{E:Gipr}
\left(
\begin{array}{c|c|c}
 ({\mathbf{N}}+ \mathbf{I})^{-1}{\mathbf{G}}_{{\vphantom{[}}_{\vphantom{[}}}& {\mathbf{G}}_{{\vphantom{[}}_{\vphantom{[}}}\mathbf{X}\ \ \ \ \ & \ \ \ \ \mathbf{N}{\mathbf{G}}_{{\vphantom{[}}}\\ \hline
 \vphantom{[}[{\mathbf{N}+ (N_k + 1)]^{-1}\mathbf{I}){\mathbf{G}}_{\vphantom{[}}}{\mathbf{X}}^{{}^{N_k}}
& {\vphantom{[}}^{{\vphantom{[}}^{\vphantom{[}}}{\mathbf{G}}_{{\vphantom{[}}_{\vphantom{[}}}{\mathbf{X}}^{{}^{N_k}} & \vphantom{[}({\mathbf{N}+N_k\mathbf{I}){\mathbf{G}}_{\vphantom{[}}}{\mathbf{X}}^{{}^{N_k}}\\ \hline
 \vphantom{[}[{\mathbf{N}+ (2N_k + 1)]^{-1}\mathbf{I}){\mathbf{G}}_{\vphantom{[}}}{\mathbf{X}}^{{}^{2N_k}}
& {\vphantom{[}}^{{\vphantom{[}}^{\vphantom{[}}}{\mathbf{G}}_{{\vphantom{[}}_{\vphantom{[}}}{\mathbf{X}}^{{}^{2N_k}} & \vphantom{[}({\mathbf{N} + 2N_k\mathbf{I}){\mathbf{G}}_{\vphantom{[}}}{\mathbf{X}}^{{}^{2N_k}}
\end{array}  % Here the subscript and superscript {\vphantom{[}} usage opens up vertical 
%              spacing.
\right)\!
\left(
\begin{array}{c}
 {\mathbf{\rho}'}_{ {{{\vphantom{R}}_{\vphantom{[}}}}_{\vphantom{[}} }\\ \hline
 {\mathbf{\mu}}_{ {{{\vphantom{R}}_{\vphantom{[}}}}_{\vphantom{[}} }\\ \hline
 {\mathbf{\nu}'}_{\vphantom{R}}^{\vphantom{[}}
\end{array}
\right)\,
 = 0\, .
\end{equation}

When there are no second order pole terms, (\ref{E:Gipr}) can be rewritten as
\begin{equation}\label{E:Gipr2}
\left(
\begin{array}{c|c}
 {\mathbf{G}}_{{\vphantom{[}}_{\vphantom{[}}}
& ({\mathbf{N}}+ \mathbf{I}){\mathbf{G}}_{{\vphantom{[}}_{\vphantom{[}}}\mathbf{X}\\ \hline
 \vphantom{[}{\mathbf{G}}_{\vphantom{[}}
{\mathbf{X}}^{{}^{N_k}}
& {\vphantom{[}}^{{\vphantom{[}}^{\vphantom{[}}}
[\mathbf{N}+ (N_k + 1)\mathbf{I}]{\mathbf{G}}_{{\vphantom{[}}_{\vphantom{[}}}
{\mathbf{X}}^{{}^{N_k}} 
\end{array}  % Here the subscript and superscript {\vphantom{[}} usage opens up vertical 
%              spacing.
\right)\!
\left(
\begin{array}{c}
 {\mathbf{\rho}'}_{ {{{\vphantom{R}}_{\vphantom{[}}}}_{\vphantom{[}} }\\ \hline
   {\mathbf{\mu}}_{ {{{\vphantom{R}}_{\vphantom{[}}}}_{\vphantom{[}} }
\end{array}
\right)\,
 = 0\,,
\end{equation}
which, by following analogous steps used to obtain (\ref{E:nuEqn2}), can be rewritten as
\begin{equation}\label{E:muEqn2}
[( N_k\mathbf{I} - {\mathbf{G}}^{-1}\mathbf{N}\mathbf{G})+ {\mathbf{X}}^{{}^{-N_k}}{\mathbf{G}}^{-1}\mathbf{N}\mathbf{G}{\mathbf{X}}^{{}^{N_k}}] \,\mu\  =\  0\ .
\end{equation}
Hence combined logarithmic and simple pole basis functions are independent if the $\mathbf{C}$ matrix introduced earlier is nonsingular.

\newpage

\appendix
\begin{center}
 \begin{Large}{\textbf{Appendix A}}\end{Large}
\end{center}

% Next set up the appendix equation labeling and numbering system:
\renewcommand{\theequation}{A-{\arabic{equation}}}
\setcounter{equation}{0}

\section*{\hfil Branch Cut and Analytic Continuation Side Issues in $\mathbb{C}$\hfil}\label{S:A}
\ \hfill  \\
\vskip -.2in 

\noindent
    This appendix addresses two distinct questions (or misperceptions) that some readers may have: (1)  Since the logarithmic point source basis functions $\psi_k$ [as given by (\ref{E:psi})] are used, and logarithmic functions generally have brach cuts in \C, does $\psi_k$ have problematic branch-cuts and, if not, why not? (2)  When only simple poles may be present in the interior region, since it seems natural to assume that the analytic continuation of a zero function is always a zero function, can one extend analytic continuation into the interior of the unit disk along various paths while assuming that $f(z) = 0$ still holds until there are areas of overlap, so that one can simply integrate around each pole separately and then directly prove the desired complex plane uniqueness results for simple poles from a direct application of the residue theorem?   This analytic continuation procedure may be tempting, because for a finite collection of potential simple poles in the interior, Theorem~2 guarantees that if $f(z) = 0$ holds in the exterior region, then $\mu_k = 0$, so this analytic continuation procedure seems to work for this case.  As far as counter examples go, one need only consider a simple circle with a uniform density constant simple pole strength and a compensating interior pole, but what about a denumerable set of separated simple poles, which is the case of prime interest here?  The reader who is not bothered by this last sort of question may simply skip the second part of this appendix that deals with this issue. 

\begin{center} 
 \hfil\\
\textbf{\underline{Branch-cut Issues}}
 \hfil\\
\end{center}

  First consider branch-cut related issues in the exterior of a unit disk, when all the source points reside inside the unit disk.  Although simple and higher order poles obviously do not have such branch cuts, at first glance it might appear that ${\psi}_k(z)$ given by (\ref{E:psi}) does have branch cuts:
\begin{equation}\label{E:psi2}
 {\psi}_k(z) \, \eq  \,\ln\,\frac1{z - z_k} - \,\ln\,\frac1{z}
\end{equation}
since, for example,
\begin{equation}\label{E:log}
\ln\, z = \ln r + i\theta\ .
\end{equation}
The truth of the matter, however, is not so straightforward and it turns out that there are no branch cuts in the exterior of the unit complex disk for ${\psi}_k(z)$.  To see this geometrically, first let $\ell_k \eq |z - z_k|$ and let $\phi_k$ equal to the angle between the positive $x$-axis direction and the vector parallel to the line segment connecting $z_k$ to $z$ [which is to say, the angle between the \R{2}\ line segment connecting the points $(x_k,\,y_k)^T$ and $(x,\,y_k)^T$ and the line segment connecting $(x_k,\,y_k)^T$ and $(x,\,y)^T$]. Then
\begin{equation}\label{E:psi3}
 {\psi}_k(z) \, \eq  \,\ln\,\frac{r}{\ell_k} + i(\theta - \phi_k)
\end{equation}
and when one draws a plane figure displaying the various relevant lines and angles, it is perfectly obvious that, for any choice of $z_k$, $\theta$ and $\phi_k$ are equal in value at two places as $z$ traces out a closed loop around the unit disk (keeping in mind that $|z| \geq 1$).
In fact, from the geometry of this figure it is clean that $\pi/2 \geq |\theta - \phi_k|$ must always hold and, hence ${\psi}_k(z)$, for each $k$. is uniquely defined for $z| \geq 1$.

  This branch-cut issue can also be settled using analysis by considering the form
\begin{equation}\label{E:psi5}
 {\psi}_k(z) \, = \,\ln\,\frac z{(z - z_k)}\, \eq  \,\ln\,\frac 1{(1 - z_k/z)}
\end{equation}
and noting that since ${\text{Re}}\,\{1 - z_k/z\} > 0$ (which follows immediately from $1 > |z_k/z| \geq |{\text{Re}}\,\{ z_k/z\}|\,\,\,$),  the absolute value of the argument of $1/(1 - z_k/z)$ is less than $\pi/2$ from elementary properties of the $\arctan$ function.

  Notice that while there are no branch-cuts for $\psi_k$ for $|z| \geq 1$, for $|z| < 1$ it is clear that there is a branch cut connecting $z_k$ and the origin.  An expansion of the form (\ref{E:complexLogs}) thus has a collection of branch cuts that form a star-like pattern in the interior of the unit disk.

  Finally, since a linear superposition of analytic function is analytic the use of $\{\psi_k\}_{k=1}^{N_k}$ as a set of basis functions, as in (\ref{E:complexLogs}) entails no branch-cut interpretational issues at all.  It is clear, however, 
that one cannot decide to treat the RHS of a $\psi_k$ fit, or expansion such as (\ref{E:complexLogs}), as a single logarithmic function using the standard properties for combining logarithms--when a linear combination of $\psi_k$'s are combined into a single composite logarithmic function brach-cut issues are arbitrarily introduced. (Since the argument of the log of some term is assumed to be between $0$ and $2\pi$, even absorbing the constant $\mu_k$ into $\psi_k$ by itself causes problems and introduces unwanted restrictions, because the complex part of $\mu_k\psi_k$ is not similarly restricted.)

\begin{center} 
 \hfil\\
\textbf{\underline{Analytic Continuation Issues}}
 \hfil\\
\end{center}

  Next, consider the analytic continuation issues, which are perhaps best addressed by consideration of counter-examples.  It may seem reasonable to argue that continuation of the zero function is a special case, especially when it done on either side of a neighborhood where it is known that, at most, one simple pole resides and it would seem that what happens outside the greater region under consideration does not matter.  Thus, suppose that an expansion of the form (\ref{E:SimplePoles}) is being considered, and that the poles are ordered such that $z_1$ is the location of the simple pole that is closest to the analytic region $(|z| > 1)$.  Consider a ``fit'' to the function $f(z) = 0$.  Then, if one analytically continues the zero function into the unit disk close to $z_1$ along  a path to one side of this pole, it would seem obvious that $f(z) = 0$ over this entire region.  On the other hand, if one starts from the same region and does the same thing on the other side of the pole then in the region of overlap between these two analytic continuations, it is obvious that $f(z) = 0$.  From this, one can conclude that the simple pole necessarily has a weight of zero: $\mu_k = 0$.  Proceeding in a like fashion to each succeeding pole, one could thus argue that $f(z) = 0$ for $|z| > 1$ implies that $\mu_k = 0$ for all $k$, which is the desired result.  Although, at each step, since there are two analytic regions of continuation that are not simply connected, and it is known that analytic continuation into regions that are not simply connected are problematic; the simple pole in question produces no branch cuts in either of these two regions and, furthermore, the $f(z)$ always agrees in the region of overlap.  No function is simpler than the zero function and one might argue that in this special case of analytic continuation there is no problem, at least in this particular instance, and, when used in this way, it leads to the right result.

To see what is wrong with this argument, consider the following.  It will only be necessary to show that one counter example exists (here some liberty will be taken and it will merely be shown that some counter-example is very likely to exist, without explicitly constructing the explicit counter-example itself).  Let $\{s_j\}_{j=1}^{\infty}$ be a sequence of points in the interior of the unit circle (excluding the origin).  Consider the set of corresponding points in the exterior of the unit circle given by $t_k = 1/s^*_j$.  Then it is well known that the sequence of values $\{f(t_j)\}_{j=1}^{\infty}$ completely characterizes the analytic function $f(z)$.  Consider a DIDACKS fit of the form
\begin{equation}\label{E:Fseq}
\varphi(z) = \sum\limits_{k=1}^N \frac{a_k}{z - s_k}\ ,
\end{equation} 
which always exist for $N < \infty$ since the associated linear system was always shown to be solvable in \cite{DIDACKSI}.
It does not seem unreasonable to assume that for some sequence of points $\{t_k\}$ and some choice of analytic function $f(z)$ that this system remains solvable as $N \rightarrow \infty$.  [For example, it is clearly perfectly acceptable to choose an $f(z)$ that has the form (\ref{E:Fseq}) itself, with some appropriate choice of $\{a_j\}_{j=1}^{\infty}$ and $\{s_j\}_{j=1}^{\infty}$, so long as $|a_j| \neq 0$ (in the argument that follows, it is important that the sum in (\ref{E:Fseq}) contains an infinite number of simple poles).]

  Now, $f(z)$ for $|z| > 1$ could have been represented by its values at some other infinite sequence of points, say $\{f(p_k)\}_{k=1}^{\infty}$, where $p_k \neq t_j$ for all $j$ and $k$.
Here let $z_k = 1/p^{*}_k$ and assume that a fit of the form (\ref{E:SimplePoles}) with $n = \infty$ also exists.
Then, for $|z| \geq 1$:
\begin{equation}\label{E:fzero}
f(z) = 0 = \sum\limits_{k=1}^{\infty} \frac{a_k}{z - s_k}\ - \ \sum\limits_{k=1}^{\infty} \frac{\mu_k}{z - z_k} \ .
\end{equation}
Since by construction, $a_j \neq 0$ for all $j$ and $\mu_k \neq 0$ for all $k$, it is clear that the strategy of analytically continuing the zero function so as to encompass an isolated simple pole must fail if \underline{any} representation of $f(z) = 0$ exists that has the form (\ref{E:fzero}); moreover, for this to occur it is only necessary that for some $a_j \neq 0$ for all $j$ and some $\mu_k \neq 0$ for all $k$, that $s_j$ and $p_k$ exist such that 
 \begin{equation}\label{E:SUMzero}
 \sum\limits_{k=1}^{\infty} \frac{a_k}{z - s_k}\ =\ \sum\limits_{k=1}^{\infty} \frac{\mu_k}{z - z_k}\ , 
\end{equation}
where (for all $j$ and $k$) $0< |s_k| < 1$, $0< |p_k| < 1$ and  $s_j \neq p_k$. 
It seems most probable that two such bounded sequences of simple poles exist.


\newpage


\begin{thebibliography}{9}
\bibitem{MatrixBook}
Dennis S. Bernstein,
\emph{Matrix Mathematics: Theory, Facts, and Formulas With Application to Linear Systems Theory},
Princeton University Press, Princeton, N.J., 2005. 
\bibitem{Buhmann}
Martin D. Buhmann,
\emph{Radial Basis Functions: Theory and Implementations},
Cambridge Monographs on Applied and Computational Mathematics,
Cambridge University Press, New York, N.Y., 2003.
\bibitem{PJDavis}
Philip J. Davis,
\emph{Introduction to Interpolation and Approximation},
Dover Publications, New York, N.Y., 1963.
\bibitem{Featherstone}
S. J. Claessens, W. E. Featherstone and F. Barthelmes,
\emph{Experiences with Point-mass Gravity Field Modeling in the Perth Region, Western Australia},
Geomatics Research Australasia, No. 75, 53--86.
\bibitem{Pick}
Milos Pick, Jan Picha and Vincenc Vyskocil,
\emph{Theory of the Earth's Gravity Field},
Elsevier Scientific Publishing Company, Amsterdam\,/\,London\,/\,New York, 1973.
\bibitem{DIDACKS}
Alan Rufty,
\emph{A Dirichlet Integral Based Dual-Access Collocation-Kernel Approach to Point-Source Gravity-Field Modeling}, SIAM Journal on Applied Mathematics, \textbf{68}, No. 1, 199--221.
\bibitem{DIDACKSI}
Alan Rufty,
\emph{Dirichlet integral dual-access collocation-kernel space analytic interpolation for unit disks: DIDACKS I}, [arxiv:math-ph/0702062].
\bibitem{DIDACKSII}
Alan Rufty,
\emph{Dirichlet-integral point-source harmonic interpolation over ${\mathbb{R}}^3$ spherical interiors: DIDACKS II}, [arxiv:math-ph/0702063].
\bibitem{DIDACKSIII}
Alan Rufty,
\emph{Closed-form Dirichlet integral harmonic interpolation-fits for real n-dimensional and complex half-space: DIDACKS III}, [arxiv:math-ph/0702064].
\bibitem{DIDACKSIV}
Alan Rufty,
\emph{A closed-form energy-minimization basis for gravity field source estimation: DIDACKS IV},
 [arxiv:math-ph/07XXXXXX].
\bibitem{Shilov}
Georgi E. Shilov,
\emph{Elementary and Complex Analysis},
Dover Publications, New York, N.Y., 1973 edition.
\bibitem{geoPMuniq}
D. Stromeyer and L. Ballani,
\emph{Uniqueness of the Inverse Gravimetric Problem for Point Mass Models},
Manuscripta Geodaetica, \textbf{9} (1984), 125--136.
\end{thebibliography}
\end{document}